\documentclass[journal]{IEEEtran}

\ifCLASSINFOpdf
  \usepackage[pdftex]{graphicx}
\else
  \usepackage[dvips]{graphicx}
\fi

\usepackage[cmex10]{amsmath}
\usepackage{amsfonts}

\usepackage{algorithm}
\usepackage{amsthm} 

\usepackage{algorithmic}
\usepackage{mdwmath}
\usepackage{mdwtab}

\usepackage{fixltx2e}

\usepackage{url}

\def\h{{\mathbf h}}

\def\br{{\mathbf r}}

\def\x{{\mathbf x}}
\def\y{{\mathbf y}}
\def\z{{\mathbf z}}

\def\D{{\mathbf D}}

\def\G{{\mathbf G}}

\def\I{{\mathbf I}}

\def\K{{\mathbf K}}

\def\R{{\mathbf R}}

\def\W{{\mathbf W}}
\def\X{{\mathbf X}}
\def\Y{{\mathbf Y}}

\def\1{{\mathbf 1}}
\def\0{{\mathbf 0}}

\renewcommand{\vec}[1]{\boldsymbol{#1}}

\newtheorem{lemma}{Lemma}

\usepackage[%
    pdftitle={Blind Identification of SIMO Wiener Systems based on Kernel Canonical Correlation Analysis},
    pdfauthor={Steven Van Vaerenbergh, Javier Via and Ignacio Santamaria},
]{hyperref}

\begin{document}
\title{Blind Identification of SIMO Wiener Systems based on Kernel Canonical Correlation Analysis}

\author{Steven~Van~Vaerenbergh,~\IEEEmembership{Member,~IEEE,}
        Javier~V\'ia,~\IEEEmembership{Senior Member,~IEEE,}\\
        and~Ignacio~Santamar\'ia,~\IEEEmembership{Senior Member,~IEEE}
\thanks{The authors are with the Advanced Signal Processing Group, Department
of Communications Engineering, University of Cantabria, Santander 39005,
Spain, e-mail: \{steven,jvia,nacho\}@gtas.dicom.unican.es.}
\thanks{Digital Object identifier 10.1109/TSP.2013.2248004}
}

\markboth{IEEE Transactions on Signal Processing, 2013 (To Appear)}{}

\maketitle

\begin{abstract}
We consider the problem of blind identification and equalization of single-input multiple-output (SIMO) nonlinear channels. Specifically, the nonlinear model consists of multiple single-channel Wiener systems that are excited by a common input signal. The proposed approach is based on a well-known blind identification technique for linear SIMO systems. By transforming the output signals into a reproducing kernel Hilbert space (RKHS), a linear identification problem is obtained, which we propose to solve through an iterative procedure that alternates between canonical correlation analysis (CCA) to estimate the linear parts, and kernel canonical correlation (KCCA) to estimate the memoryless nonlinearities. The proposed algorithm is able to operate on systems with as few as two output channels, on relatively small data sets and on colored signals. Simulations are included to demonstrate the effectiveness of the proposed technique.
\end{abstract}


\IEEEpeerreviewmaketitle

\section{Introduction} \label{sec:intro}
\IEEEPARstart{B}{lind} identification and equalization have been active research topics during the last decades. In digital communications, blind methods allow channel identification or equalization without the need to send known training signals, thus saving bandwidth. While a lot of attention has gone to the analysis of linear systems, many real-life systems exhibit nonlinear characteristics. As a result, the field of nonlinear system identification has been studied for many years and still remains a very active research area \cite{Billings80,Sjoberg95nonlinearblackbox,Nelles00,giannakis01bibliography}.

In contrast to linear systems, which can be identified uniquely by their impulse response, there does not exist a corresponding canonical representation for all nonlinear systems. Hence, different approaches are followed to parameterize different subclasses of nonlinear systems, including descriptions such as Volterra \cite{Schetzen80} and polynomial \cite{mathews00_polynomial} systems. While these techniques allow for adequate representations of many nonlinear systems, the number of parameters they require becomes excessive for high degrees of nonlinearity or high input dimensionality. Therefore, several authors have considered approximating the unknown nonlinear systems as simplified block-based models, including Wiener systems \cite{billings77wiener,Pawlak07}, which comprise a cascade of a linear filter and a memoryless nonlinearity; Hammerstein systems \cite{narendra66hammerstein,billings79hammerstein}, which correspond to the inverse configuration; and combinations of both \cite{Bai98,Bershad01,debrabanter2008wienerhammerstein}. We will focus on Wiener systems, which, despite their simplicity, have been used successfully in applications including biomedical engineering \cite{westwick98nonlinear}, control systems \cite{greblicki04wiener}, digital satellite communications \cite{feher83}, digital magnetic recording \cite{sands93}, optical fibre communications \cite{Kawakami07_wiener} and chemical processes \cite{pajunen92wiener}.

A considerable number of techniques have been proposed in recent years to tackle the problem of supervised identification of Wiener systems, both single-input single-output (SISO) systems \cite{Pawlak07,cousseau07wiener, vanvaerenbergh08_kcca_wiener_id, giri2010block} and multiple-input multiple-output (MIMO) systems \cite{westwick1996mimowiener, janczak2007mimowiener, fan2009mimowiener}. Nevertheless, relatively little research has been conducted on the blind identification problem. For SISO Wiener systems, some blind identification techniques have been proposed that make assumptions on the input signal statistics (see \cite[Part IV]{giri2010block}). In particular, in \cite{Gomez07_blind_wiener,vanbeylen09blindwiener} the input signal is required to be independent and identically distributed (i.i.d.) and Gaussian. A less restrictive approach was followed by Taleb et al. in \cite{TalebJutten01_blindwiener}, where the input signal is only required to be i.i.d.

The problem of blind identification of nonlinear single-input multiple-output (SIMO) systems has also been addressed, although only for the class of Volterra models. The SIMO model can be obtained for instance by measuring a single source using a sensor array. In \cite{giannakis97_blind_eq_volterra} it was shown that a finite impulse response (FIR) linear filter can perform zero-forcing (ZF) equalization of SIMO Volterra systems under certain conditions. In \cite{valcarce01_blind_eq_nl} a different technique was proposed for blind equalization of SIMO Volterra models, based on second-order statistics (SOS), that improved several aspects of \cite{giannakis97_blind_eq_volterra}, including computational complexity and robustness. Both methods require at least three output channels to operate. 

\begin{figure}[t]
\centering %
\includegraphics{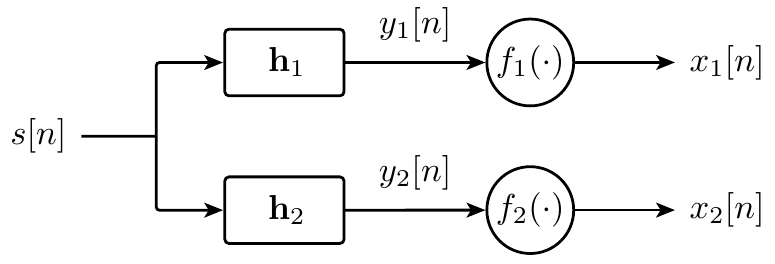}
%
\caption{The block diagram of a SIMO system consisting of two Wiener systems. The proposed method aims to identify the entire system and to estimate $s[n]$ given only $x_1[n]$ and $x_2[n]$.}
\label{fig:blindideqwiener_system_simo2}
\end{figure}

In this paper we will focus on the blind identification and equalization of SIMO Wiener systems, as depicted in Fig.~\ref{fig:blindideqwiener_system_simo2}. We propose a blind technique that requires looser restrictions than blind SISO techniques, and that is able to operate with two or more output channels. The proposed identification approach is based on a well-known technique in blind identification of linear SIMO systems \cite{xu95_ls,Via06_effectiveorder}. For SIMO Wiener systems the blind identification problem is more challenging, as it includes nonlinearities. By drawing on the framework of kernel methods, however, the problem can be linearized. 
Some preliminary results of the proposed method were presented in \cite{vanvaerenbergh08_kccablind}. We extend these results with an identifiability analysis, a general formulation for multiple outputs, a formulation that exploits identical nonlinearities in each channel and more exhaustive numerical experiments.

The rest of this paper is organized as follows: Section \ref{sec:simolinear} gives a brief review of blind identification methods for linear SIMO systems based on SOS. This scenario is extended to SIMO Wiener systems with two output channels in Section \ref{sec:simowiener} and generalized to systems with multiple outputs in Section \ref{sec:ext}. Section \ref{sec:experiments} contains a series of numerical experiments, and the main conclusions of this work are presented in Section \ref{sec:conclusions}.

Throughout this paper the following notation is used: Scalar variables are denoted as lowercase letters, $x$, and vectors as boldface lowercase letters, $\x$, defined as column vectors. Matrices are indicated by boldface uppercase letters, such as $\X$. Square brackets denote the instance of any variable at time $n$, or the $n$-th element of a matrix or vector, $x[n]$, and a hat denotes an estimate of a variable, $\hat\x$.

\section{Blind Identification of Linear SIMO Systems} \label{sec:simolinear}

We start by reviewing the basic blind identification problem of a linear system with two outputs. The extension to multiple outputs is straightforward, as shown in \cite{xu95_ls} and
\cite{Via06_effectiveorder}. The signals used in this paper are real, although the proposed methods can be easily extended for complex signals.

Consider a system that consists of two linear channels $\h_1$ and $\h_2$ that share the same zero-mean input signal, $s[n]$, as depicted in Fig.~\ref{fig:blindideqlinear_simo2_eq}. Assuming FIR channels, the output of the $i$-th channel can be written as
\begin{equation*}
x_i[n] = \sum_{l=0}^{L-1} h_i[l] s[n-l] = h_i[n] \ast s[n],
\end{equation*}
where $\h_i=\left[h_i[0],\dots,h_i[L-1]\right]^T$ denotes the impulse response vector of the $i$-th channel, $L$ is the maximal channel length (which is assumed to be known), and $h_i[n]\ast s[n]$ is the convolution between $\h_i$ and the input signal $s[n]$. This system can be obtained for instance by oversampling a single linear channel when the source signal has some excess bandwidth, which is the bandwidth occupied by the signal beyond the Nyquist frequency $1/2T_s$ (see \cite[Section 9.2.1]{proakis2007}).

\begin{figure}[t]
\centering %
\includegraphics{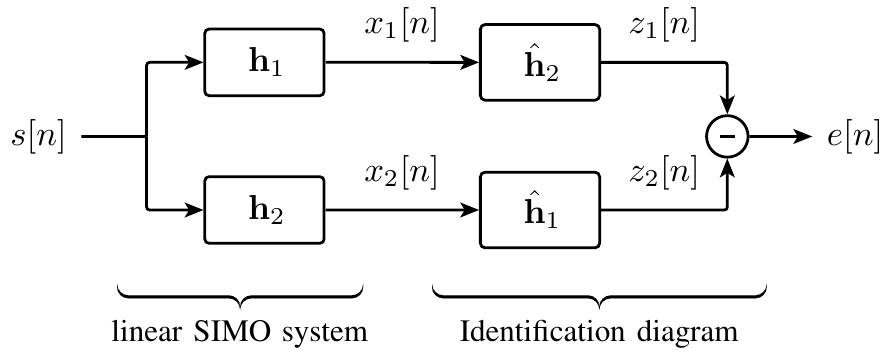}
%
%
%
\caption{A linear SIMO system and the corresponding blind identification diagram. If $\hat\h_1=\h_1$ and $\hat\h_2=\h_2$, the output error $e[n]$ will be zero.} \label{fig:blindideqlinear_simo2_eq}
\end{figure}

The identification method presented by Xu et al. in \cite{xu95_ls}, which is closely related to linear prediction, exploits the commutativity of the convolution, in particular
\begin{equation*}
h_2[n] \ast (h_1[n] \ast s[n]) = h_1[n] \ast (h_2[n] \ast s[n]).
\end{equation*}
This property inspired the design of the identification diagram shown in Fig.~\ref{fig:blindideqlinear_simo2_eq}, which allows to find estimates of the channels,
$\hat{\h}_1$ and $\hat{\h}_2$, by minimizing the following cost function
\begin{equation}
J = \frac{1}{2}\sum_{n=1}^N |z_1[n] - z_2[n]|^2 = \frac{1}{2}\sum_{n=1}^N |e[n]|^2 \label{eq:blindideq_simo_linear1},
\end{equation}
with respect to $\hat{\h}_1$ and $\hat{\h}_2$, where
\begin{equation*}
z_1[n] = \hat h_2[n] \ast x_1[n] = \hat h_2[n] \ast (h_1[n] \ast s[n]),
\end{equation*}
and $z_2[n]$ is constructed in a similar fashion.

In order to solve this minimization problem, we define the data matrix
\begin{equation*}
\X_i =
\begin{bmatrix}
x_i[n+L-1] & \cdots & x_i[n]\\
\vdots & \ddots & \vdots\\
x_i[n+N-1] & \cdots & x_i[n+N-L]\\
\end{bmatrix}, \quad i=1,2.\label{eq:blindideqdatamatrix}
\end{equation*}
By denoting the estimate of the channel impulse response vectors as
\begin{equation*}
\hat \h_i = \left[  \hat h_i[0], \dots, \hat h_i[L-1] \right]^T,
\end{equation*}
it can be easily verified that in a noiseless case the solution should satisfy
\begin{equation}
\X_1 \hat \h_2 = \X_2 \hat \h_1, \label{eq:tong_noiseless}
\end{equation}
as illustrated in the identification diagram of Fig.~\ref{fig:blindideqlinear_simo2_eq}. Correct identification is guaranteed when the channels $\h_i$ do not share any common zeros and the linear complexity of the input sequence is sufficiently high \cite{xu95_ls}. For real-world signals this is generally satisfied, see [31].

In case the outputs $x_i[n]$ are corrupted by additive noise, Eq.~\eqref{eq:tong_noiseless} cannot be fulfilled in general, and the optimal filters $\hat\h_1$ and $\hat\h_2$ need to be determined by solving an optimization problem. In order to avoid the zero-solution $\hat \h_i = \0$, either the norm of the filters $\hat \h_i$ or the norm of the output signal $\X_i \hat \h_j$ is typically fixed. A restriction on the filter norm was used in \cite{xu95_ls} to develop a least-squares (LS) method. With this restriction, the minimization problem \eqref{eq:blindideq_simo_linear1} becomes
\begin{equation*}
\begin{aligned}
& \underset{\hat{\h}_1,\hat{\h}_2}{\text{minimize}} & &  \frac{1}{2} \left\| \X_1 \hat \h_2 - \X_2  \hat \h_1 \right\|^2 \\
& \text{subject to} & & \| \hat \h_1 \|^2 + \| \hat \h_2 \|^2 = 1.
\end{aligned} \label{eq:blindideq_xu_cost1}
\end{equation*}
Its solution is obtained by solving the eigenvalue problem
\begin{equation}
\begin{bmatrix}
\X_2^T \X_2 & -\X_2^T \X_1\\
-\X_1^T \X_2 & \X_1^T \X_1
\end{bmatrix}
\hat \h = \rho \hat \h, \label{eq:blindideqls_tong}
\end{equation}
in which $\hat \h = [\hat \h_1^T, \hat \h_2^T]^T$ is found as the eigenvector corresponding to the smallest eigenvalue.

If a constraint is applied on the output signal energy (as in \cite{Via06_effectiveorder}), the following optimization problem is obtained
\begin{equation}
\begin{aligned}
& \underset{\hat{\h}_1,\hat{\h}_2}{\text{minimize}} & &  \frac{1}{2} \left\| \X_1 \hat \h_2 - \X_2  \hat \h_1 \right\|^2 \\
& \text{subject to} & & \| \X_1 \hat \h_2 \|^2 = \| \X_2 \hat \h_1 \|^2 = 1.
\end{aligned} \label{eq:cca1}
\end{equation}
Problem \eqref{eq:cca1} is a canonical correlation analysis (CCA) problem, whose solution is given by the principal eigenvector of the following generalized eigenvalue problem (GEV) (see \cite{Hotelling36}).
\begin{equation}
\begin{bmatrix}
\0 & \X_2^T \X_1\\
\X_1^T \X_2 & \0
\end{bmatrix}
\hat \h = \rho
\begin{bmatrix}
\X_2^T \X_2 & \0\\
\0 & \X_1^T \X_1
\end{bmatrix}
\hat \h. \label{eq:blindideqgev_cca}
\end{equation}
Note that both the LS and CCA-based algorithms require knowledge of the maximum channel length $L$. More generally, the following assumptions are required in order to guarantee identifiability \cite{xu95_ls}:
\begin{enumerate}
\item[A1.] The linear channels $\h_i$ are coprime, i.e. they do not share any common zeros.
\item[A2.] The linear complexity of the input signal is at least $2L+1$, where $L$ is the maximum length of the linear channels.
\end{enumerate}

Once the channels $\hat \h_1$ and $\hat \h_2$ have been estimated by solving either one of the eigenvector problems \eqref{eq:blindideqls_tong} or \eqref{eq:blindideqgev_cca}, system equalization can be performed by applying the zero-forcing (ZF) or the minimum mean-square error (MMSE) algorithm. For the proposed technique we choose to work with the constraint based on the output signal energy, and its corresponding CCA formulation \eqref{eq:blindideqgev_cca}, since it reduces the noise enhancement problem, especially in the case of colored signals or a small number of observations \cite{via07cca_eq}.

\section{Blind Identification and Equalization of SIMO Wiener Systems}
\label{sec:simowiener}

The problem of interest consists of nonlinear SIMO system identification, in which each channel is modeled as a Wiener system. This model can be obtained by using a sensor array at the receiving end, given that each individual sensor allows to be represented as a Wiener system. In accordance to the nomenclature used in the literature we call this system a SIMO Wiener system. Fig.~\ref{fig:blindideqwiener_system_simo2} displays a system with two outputs, as encountered for instance in \cite{Chen95_Parallel_Nonlinear_Systems}. The output of the $i$-th channel is obtained as
\begin{equation*}
x_i[n] = f_i\left(\sum_{l=0}^{L-1} h_i[l] s[n-l]\right).
\end{equation*}
We will restrict the nonlinearities $f_i(\cdot)$ to be monotonic and invertible in this work, since the proposed identification method is based on estimating the inverse nonlinearities. This restriction is fulfilled in many practical scenarios, for instance when the nonlinearities are modeled as saturating nonlinearities (which is the case for saturating amplifiers, limit switch devices in mechanical systems and overflow valves among others --- see examples in \cite{giri2010block}).

Before describing the details of the proposed method we discuss the identifiability conditions of this system.

\subsection{Blind identifiability} \label{sec:identifiability}

We start by pointing out some ambiguities that need to be taken in mind when identifying a SIMO Wiener system. 
Throughout this discussion it is understood that any identification solution is only given up to a set of scalar constants, which represent scalings of its unknown internal signals $y_i[n]$ and its source signal $s[n]$. Furthermore, the linear channels of the SIMO Wiener system should be of length $L>1$. A system with $L=1$ represents a degenerate case, as it is impossible to identify its nonlinearities: For instance, any monotonic transformation $\theta(\cdot)$ of its source signal would allow to construct a different SIMO Wiener system that has different nonlinearities, $f_i(\theta^{-1}(\cdot))$, while having the same output signals.

An important observation is that the described system is \emph{not} identifiable in general when the input signal is of finite length. 
In order to prove this statement we will make use of the concept of \emph{amplitude order}, based on order statistics: Given a sequence of samples, we define the amplitude order as the order of these samples when they are sorted ascendingly, where samples with identical values are given the same order. For instance, the amplitude order of the sequence $[1, 4, 3, 3]$ is $[1, 3, 2, 2]$. An interesting property is that if two sequences have the same amplitude order there exists a monotonic function $\theta(\cdot)$ that transforms one sequence into the other one.
\begin{lemma}
A SIMO Wiener system with a monotonic invertible nonlinearity is not identifiable in general if its input signal $s[n]$ is of finite length.
\label{lemma:nonid}
\end{lemma}
\begin{proof}
We first show that a SISO Wiener system is not identifiable in general for finite $N$. The proof is given by a simple counterexample. Denote by $\h$ the linear channel of a given Wiener system, by $f(\cdot)$ its nonlinearity and by $y[n]$ its intermediate signal. Consider $\delta$ to be the minimal distance between any two consecutive ordered samples $y[n]$. Since the input signal $s[n]$ to this system is of finite length, we can assume that $\delta$ will be small but non-null.

Now consider an alternative Wiener system with input signal $\hat s[n] = s[n] + \epsilon[n]$, linear channel $\hat\h + \vec\nu$ and intermediate signal $\hat y[n]$, where $\epsilon[n]$ represents a perturbation signal and $\vec\nu$ is a channel perturbation that is not a scaling of $\h$. It is clear that by choosing the perturbations small enough w.r.t. $\delta$, but not zero, the amplitude order of each $\hat y[n]$ can be made identical to the one of its corresponding $y[n]$. Therefore, $\hat y[n]$ and $y[n]$ can be transformed one into the other through a function $\hat y[n] =  \theta(y[n])$. By choosing the nonlinearity of the alternative Wiener system to be $f(\theta^{-1}(\cdot))$ a Wiener system is obtained that is different from the given system but whose output sequence is identical. Hence, the given Wiener system is not uniquely identifiable.

The previous counterexample can be applied to each branch of a SIMO Wiener system independently. Therefore, if no additional assumptions are made, a SIMO Wiener system is not uniquely identifiable when its input signal has finite length.
\end{proof}

While Lemma~\ref{lemma:nonid} may seem discouraging, 
it requires to be put in a practical perspective. As the previous example shows, the norm of the allowed perturbations depends on the differences $\delta$ between consecutive ordered samples. If the length $N$ of the input signal grows and the nonlinearities become completely excited in their ranges, it is reasonable to assume that the $\delta$-values will shrink. As a result, the norm of the allowed perturbations will shrink as well, and hence the identification error reduces. In the limit case of $N\rightarrow \infty$ the system becomes completely identifiable. Note that the reason the system is not identifiable in theory for finite $N$ is the unrestricted flexibility of the nonlinearities, represented by $\theta(\cdot)$ in the example. If this flexibility is somehow limited, though, identifiability becomes possible. In a practical scenario this is generally true, 
as can be motivated by the principle of parsimony. Therefore, as long as the nonlinearities are sufficiently smooth, it is possible to identify a SIMO Wiener system using only a finite number of samples. We will model the nonlinearities as non-parametric \emph{kernel expansions} (see Section~\ref{sec:km}), which allow to impose different degrees of smoothness on the nonlinearities without limiting their shape to any particular model.

Based on the previous discussion we can formulate a set of assumptions, in addition to A1 and A2,  that guarantee identifiability in most practical situations:
\begin{itemize}
\item[A3.] $L>1$;
\item[A4.] The nonlinearities are invertible and monotonic;
\item[A5.] $N\gg1$ and each $f_i(\cdot)$ is sufficiently excited in its range.
\end{itemize}
Appropriate values of $N$ depend on each scenario individually. Specifically, the smoother the nonlinearities of the system, the lower $N$ can be. As we will see in Section~\ref{sec:experiments}, relatively small sample sizes are sufficient in practice.

\begin{figure*}[tp]
\centering %
\includegraphics{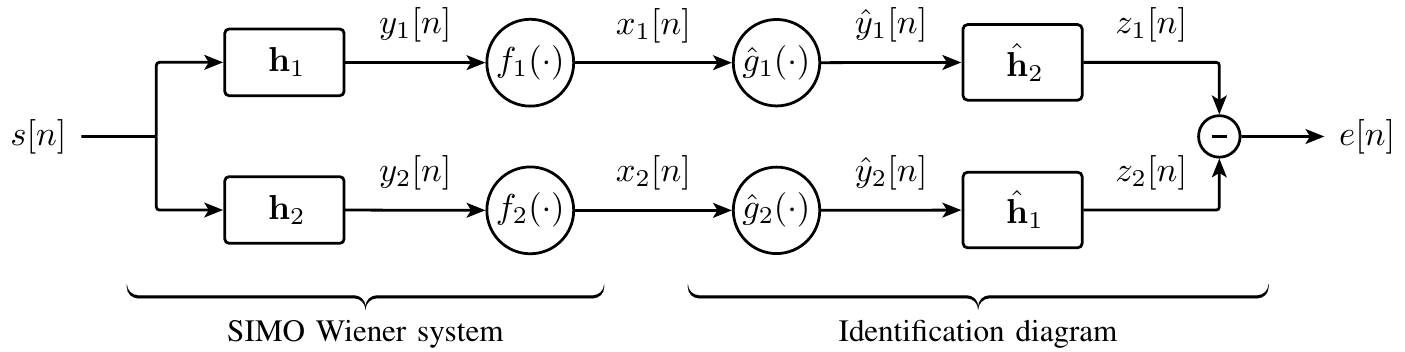}
%
%
%
\caption{A SIMO system consisting of two Wiener subsystems, followed by the proposed identification diagram in the form of a MISO Hammerstein system.} \label{fig:blindideqwiener_system_simo2_eq}
\end{figure*}


\subsection{Outline of the proposed method} \label{sec:outline}
We now describe the proposed blind identification method, starting with the two-channel Wiener system.
The proposed identification diagram, which has the structure of a multiple-input single-output (MISO) Hammerstein system, is pictured in Fig.~\ref{fig:blindideqwiener_system_simo2_eq}.
In particular, since the nonlinearities $f_i(\cdot)$ of the SIMO Wiener system are assumed to be invertible, they can be canceled out by applying the inverse nonlinearities $g_i(\cdot) = f_i^{-1}(\cdot)$ to the system outputs $x_i[n]$. 
If the nonlinearities were known, the problem would reduce to identifying the linear channels $\h_1$ and $\h_2$, which is achieved by applying either one of the discussed linear techniques. However, since the nonlinearities $f_i(\cdot)$ are also unknown, they need to be estimated jointly with the linear part.

Similarly to the linear scenario, we define the cost function
\begin{equation}
J = \frac{1}{2} \sum_{n=1}^N |z_1[n] - z_2[n]|^2 = \frac{1}{2} \sum_{n=1}^N |e[n]|^2, \label{eq:cost_basic}
\end{equation}
which uses the identification diagram outputs, defined as
\begin{equation}
z_1[n] = \sum_{l=0}^{L-1} \hat h_2[l] \hat g_1\left(x_1[n-l]\right), \label{eq:output_general}
\end{equation}
and equivalent for $z_2[n]$. The minimization of Eq.~\eqref{eq:cost_basic} represents a nonlinear optimization problem, which is generally hard to solve. In order to avoid trivial solutions such as the zero-solution and overfit solutions caused by an excessive flexibility of the nonlinearities, the solutions will require to be restricted in several ways. We will resort to the framework of kernel methods to implement these restrictions and to linearize the problem.

\subsection{Kernel methods} \label{sec:km}

Kernel methods are powerful nonlinear techniques based on a nonlinear transformation of the data $\x$ into a high-dimensional reproducing kernel Hilbert space (RKHS), in which it is more likely that the transformed data $\Phi(\x)$ is linearly separable. In this feature space, inner products can be calculated by using a positive definite kernel function satisfying Mercer's condition
\cite[Chapter 5]{Vapnik95nature}: $\kappa(\x, \x') = \langle \Phi(\x), \Phi(\x') \rangle$. This simple and elegant idea, also known as the ``kernel trick'', allows us to perform inner-product based algorithms implicitly in feature space by replacing all inner products by kernels. The solution of the resulting linear problem in feature space then corresponds to the solution of the nonlinear problem in the original space. Common kernel-based algorithms include support vector machines (SVM) \cite[Chapter 5]{Vapnik95nature} and kernel principal component analysis (KPCA) \cite{scholkopf98_kev}.

Thanks to the Representer theorem \cite{scholkopf01_representer}, a large class of optimization problems in RKHS have solutions that can be expressed as kernel expansions in terms of the available data. Specifically, it allows us to model a nonlinearity $g(\cdot)$ as
\begin{equation}
y = g(\x) = \sum_{n=1}^N \alpha[n] \kappa(\x, \x[n]). \label{eq:kernel_exp_g}
\end{equation}
where $\{\x[n]|n=1,\dots,N\}$ are the training data. It has been shown that this expansion acts as a universal approximator \cite{micchelli2006universal} for sufficiently rich kernels such as the Gaussian kernel,
\begin{equation*}
\kappa(\x,\x') = \exp(-\|\x-\x'\|^2/2\sigma^2). \label{eq:gaussian}
\end{equation*}
For a given set of $N$ input-output data pairs $(\x[n],y[n])$, Eq.~\eqref{eq:kernel_exp_g} can be written in matrix form as
\begin{equation}
\y = \K \vec\alpha, \label{eq:kernel_exp_matrix}
\end{equation}
where $\y = [y[1],\dots,y[N]]^T$, $\vec\alpha = [\alpha[1],\dots,\alpha[N]]^T$, and $\K \in \mathbb{R}^{N\times N}$ is the kernel matrix with elements
\begin{equation}
K[i,j] = \kappa(\x[i],\x[j]). \label{eq:kernelmatrix}
\end{equation}
As we will see in the sequel, smoothness constraints can be imposed on the represented nonlinearity by restricting the norm of $\vec\alpha$. First, though, we outline the proposed optimization problem using kernel expansions to represent the estimated nonlinearities.

\subsection{Proposed optimization problem}

Consider the output $z_1[n]$ of the first branch of the proposed identification scheme of Fig.~\ref{fig:blindideqwiener_system_simo2_eq}. By introducing the kernel expansion \eqref{eq:kernel_exp_g} into Eq.~\eqref{eq:output_general}, it can be written as
\begin{equation}
z_1[n] = \sum_{l=0}^{L-1} \sum_{m=1}^{N} \hat h_2[l] K_1[n-l,m] \hat\alpha_1[m], \label{eq:blindideqoutput}
\end{equation}
where $\K_1$ is the kernel matrix of the available data $x_1[n]$ of this branch.
The entire output vector $\z_1 = [z_1[1], z_1[2], \dots, z_1[N]]^T$ can thus be written as
\begin{equation*}
\z_1 = \bar\K_1 \br_2,
\end{equation*}
in which the elements of $\bar\K_1 \in \mathbb{R}^{N\times(LM)}$ are defined as
\begin{equation*}
\bar K_1[n,lM+m] = K_1[n-l,m],
\end{equation*}
and $\br_2$ represents the Kronecker product of $\hat \h_2 = [\hat h_2[0],\dots,\hat h_2[L-1]]^T$ and $\hat{\vec\alpha}_1 = [\hat\alpha_1[1],\dots,\hat\alpha_1[N]]^T$,
\begin{equation*}
\br_2 = \hat\h_2 \otimes \hat{\vec\alpha}_1.
\end{equation*}
After obtaining a similar expression for the second output channel, i.e. $\z_2 = \bar\K_2\br_{1}$, the linear optimization problem \eqref{eq:cca1} is extended for SIMO Wiener systems as
\begin{equation}
\begin{aligned}
& \underset{\br_{1},\br_{2},\hat\h_1,\hat\h_2,\hat{\vec\alpha}_1,\hat{\vec\alpha}_2}{\text{minimize}} & & \|\bar\K_1 \br_{2} - \bar\K_2\br_{1}\|^2 \\
& \text{subject to} & & \|\bar\K_1 \br_{2}\|^2 = \|\bar\K_2 \br_{1}\|^2 = 1\\
& & & \br_{2} = \hat\h_2 \otimes \hat{\vec\alpha}_1\\
& & & \br_{1} = \hat\h_1 \otimes \hat{\vec\alpha}_2.
\end{aligned}
\label{eq:optim2}
\end{equation}
For simplicity, we denote this problem as
\begin{equation}
\begin{aligned}
& \underset{\hat\h_1,\hat\h_2,\hat{\vec\alpha}_1,\hat{\vec\alpha}_2}{\text{minimize}} & & \left\| \z_1 - \z_2 \right\|^2 \\
& \text{subject to} & & \|\z_1\|^2 = \|\z_2\|^2 = 1,
\end{aligned} \label{eq:optim2_simple}
\end{equation}
where we have omitted the trivial dependency of $\z_1$ and $\z_2$ on $\hat\h_1$, $\hat\h_2$, $\hat{\vec\alpha}_1$ and $\hat{\vec\alpha}_1$, see Eq.~\eqref{eq:blindideqoutput}.


\subsection{Alternating optimization procedure}

The optimization problem \eqref{eq:optim2_simple} is not convex and generally hard to solve. However, if $\hat{\vec\alpha}_1$ and $\hat{\vec\alpha}_2$ were available, this problem would reduce to the easier problem \eqref{eq:cca1}. Equivalently, if $\hat \h_1$ and $\hat \h_2$ were known, a similar reduction would lead to another optimization problem of the form of \eqref{eq:cca1} that would yield solutions for $\hat{\vec\alpha}_1$ and $\hat{\vec\alpha}_2$. This suggests an iterative scheme that alternates between updating the estimates of the linear channels $\hat\h_i$ and the nonlinearity estimates $\hat{\vec\alpha}_i$. Convergence is guaranteed because each update may either decrease or maintain the cost \cite{bezdek03alternating,stoica04cyclic}.

\subsubsection{Iteration 1: given $\hat{\vec\alpha}_i$, obtain $ \hat \h_i$}

If estimates of $\hat{\vec\alpha}_1$ and $\hat{\vec\alpha}_2$ are given, the output $z_{1}[n]$ is
\begin{equation}
z_{1}[n] = \sum_{l=0}^{L-1} \hat h_2[l] \hat y_1[n-l], \label{eq:blindideqoutput_alpha_known}
\end{equation}
in which the elements $\hat y_1[n]$ are obtained with Eq.~\eqref{eq:kernel_exp_g}. In matrix form, Eq.~\eqref{eq:blindideqoutput_alpha_known} becomes
\begin{equation*}
\z_{1} = \hat \Y_1 \hat \h_2,
\end{equation*}
where the $n$-th row of the matrix $\hat \Y_1$ contains the elements from $\hat y_1[n]$ until $\hat y_1[n+L-1]$. This allows us to rewrite the minimization problem \eqref{eq:optim2_simple} as
\begin{equation}
\begin{aligned}
& \underset{\hat\h_1, \hat\h_2}{\text{minimize}} & & \| \hat \Y_1 \hat \h_2 - \hat \Y_2 \hat \h_1 \|^2\\
& \text{subject to} & & \| \hat \Y_1 \hat \h_2 \|^2 = \| \hat \Y_2 \hat \h_1 \|^2 = 1. \label{eq:blindideqcca_h}
\end{aligned}
\end{equation}
This problem is equivalent to the CCA problem \eqref{eq:cca1}, whose solution is found by solving the GEV \eqref{eq:blindideqgev_cca} \cite{Hardoon04_CCA}.
\subsubsection{Iteration 2: given $\hat \h_i$, obtain $\hat{\vec\alpha}_i$}

If estimates of $\hat \h_1$ and $\hat \h_2$ are given, we obtain
\begin{equation}
z_{1}[n] = \sum_{m=1}^{N} W_{1}[n,m] \hat\alpha_1[m], \label{eq:blindideqoutput_h_known}
\end{equation}
where the auxiliary variable
\begin{equation}
W_{1}[n,m] = \sum_{l=0}^{L-1} \hat h_2[l] K_1[n-l,m] \label{eq:def_w2}
\end{equation}
is introduced. In matrix form, Eq.~\eqref{eq:blindideqoutput_h_known} can be written as
\begin{equation*}
\z_{1} = \W_{1} \hat{\vec\alpha}_1,
\end{equation*}
with $\W_1 \in \mathbb{R}^{N\times N}$. By doing so, the minimization problem \eqref{eq:optim2_simple} becomes
\begin{equation}
\begin{aligned}
& \underset{\hat{\vec\alpha}_1, \hat{\vec\alpha}_2}{\text{minimize}} & & \| \W_{1} \hat{\vec\alpha}_1 - \W_{2} \hat{\vec\alpha}_2 \|^2 \\
& \text{subject to} & & \| \W_{1} \hat{\vec\alpha}_1 \|^2 = \| \W_{2} \hat{\vec\alpha}_2 \|^2 = 1. \label{eq:blindideqkcca_alpha}
\end{aligned}
\end{equation}
which establishes a kernel CCA problem that accounts for the estimation of the nonlinearities $g_i(\cdot)$. The solution is found by solving the associated GEV, which is similar to \eqref{eq:blindideqgev_cca}.

\subsection{Computational issues and regularization}

We now discuss some of the computational issues that need to be solved to guarantee that the proposed procedure performs correctly and efficiently.

\subsubsection{Low-rank approximations}

Solving a GEV generally requires cubic time and memory complexity in terms of the involved matrix sizes, i.e. $\mathcal{O}(N^3)$. Accordingly, if the training set is large, the GEV for Eq.~\eqref{eq:blindideqkcca_alpha} will pose very large computational requirements. Interestingly however, kernel matrices usually have a quickly decaying spectrum \cite{Williams01_nystrom,smola2000sparse}, which allows to approximate them reliably by a low-rank decomposition of the form
\begin{equation}
\K_i \cong \G_i \G_i^T, \label{eq:decomposition}
\end{equation}
where $\G \in \mathbb{R}^{N\times M}$ and $M\ll N$. Following Eq.~\eqref{eq:kernel_exp_matrix}, we obtain
\begin{equation}
\hat\y_i = \G_i \hat{\vec\alpha}_i, \label{eq:yG}
\end{equation}
where $\hat{\vec\alpha}_i$ now contains a reduced set of $M$ expansion coefficients.

The auxiliary variables defined in Eq.~\eqref{eq:def_w2} are replaced by
\begin{equation}
W_{1}[n,m] = \sum_{l=0}^{L-1} \hat h_2[l] G_1[n-l,m]. \label{eq:def_w2_G}
\end{equation}
The new matrices $\W_i$ have dimensions $N\times M$, which reduces the complexity of the GEV for Eq.~\eqref{eq:blindideqkcca_alpha} to $\mathcal{O}(M^3)$. 
Several methods have been proposed to retrieve suitable kernel matrix decompositions in $\mathcal{O}(NM^2)$ time, most notably Nystr\"om approximation \cite{Williams01_nystrom}, sparse greedy approximations \cite{smola2000sparse}, and incomplete Cholesky decomposition (ICD) \cite{Bach02kernelica}. We will use the latter in the simulations.

\subsubsection{Data centering}

An important requirement of CCA is that the input data be centered. For KCCA, this translates into the need to center the data in feature space \cite{Bach02kernelica}. While it is hard to remove the mean explicitly from the transformed data $\Phi(x[n])$, the approximate kernel matrix $\G \G^T$ can be centered easily in feature space by performing the transformation
\begin{equation}
\G \leftarrow \left(\I - \frac{1}{N}\1 \right) \G, \label{eq:centering}
\end{equation}
where $\I$ is the unit matrix and $\1$ is an $N\times N$ all-ones matrix. This operation simply removes the column means from $\G$, and can thus be implemented without explicitly calculating any $N\times N$ matrices \cite{vanvaerenbergh2010kernel}.

\subsubsection{Regularization}

If any of the matrices $\W_i$ is invertible, the GEV~\eqref{eq:blindideqkcca_alpha} does not yield a useful solution as it allows to find perfect correlation between any two data sets.
This is a standard issue in KCCA that stems from the unbounded flexibility of the nonlinearities, which is a property we seek to avoid (see Section~\ref{sec:identifiability}). A straightforward fix is to regularize the flexibility of the projections $\hat{\vec\alpha}_i$ by penalizing their norms, as follows \cite{Bach02kernelica,Hardoon04_CCA}:
\begin{equation}
\begin{aligned}
& \underset{\hat{\vec\alpha}_1, \hat{\vec\alpha}_2}{\text{minimize}} & & \| \W_{1} \hat{\vec\alpha}_1 - \W_{2} \hat{\vec\alpha}_2 \|^2 \\
& \text{subject to} & & \| \W_{1} \hat{\vec\alpha}_1 \|^2 + c \|\hat{\vec\alpha}_1\|^2 = 1\\
& & & \| \W_{2} \hat{\vec\alpha}_2 \|^2 + c \|\hat{\vec\alpha}_2\|^2 = 1, \label{eq:blindideqkcca_alpha_reg}
\end{aligned}
\end{equation}
where $c$ is a small regularization factor. This yields the following GEV, which combines low-rank approximation and regularization
\begin{multline}
\begin{bmatrix}
\0 & \W_1^T \W_2\\
\W_2^T \W_1 & \0
\end{bmatrix}
\begin{bmatrix}
\hat{\vec\alpha}_1\\
\hat{\vec\alpha}_2
\end{bmatrix}
\\ = \rho
\begin{bmatrix}
\W_1^T \W_1 + c\I & \0\\
\0 & \W_2^T \W_2 +c\I
\end{bmatrix}
\begin{bmatrix}
\hat{\vec\alpha}_1\\
\hat{\vec\alpha}_2
\end{bmatrix}
\label{eq:blindideqgev_cca_reg}
\end{multline}

\subsection{Initialization and algorithm overview} \label{sec:overview}

Analogously to many other iterative techniques, the proposed cyclic minimization algorithm could suffer from local minima. In practice, local minima can be avoided by means of a proper initialization technique. A straightforward initialization consists in estimating the initial nonlinearities as the identity function $g_i(x) = x$, and obtaining the initial estimate of the linear channels $\hat \h_i$ by solving the linear CCA problem \eqref{eq:cca1} for the system outputs $x_i[n]$. 

In case a more accurate initialization is required, the optimization problem \eqref{eq:optim2} can be solved directly with respect to $\br_1$ and $\br_2$, after making the necessary modifications to take into account regularization and low-rank decompositions. The Kronecker structure can be forced a-posteriori onto the estimates of $\br_1$ and $\br_2$ by applying singular value decomposition (SVD) on them (specifically, on the $M\times L$ matrices that are obtained by ordering their elements column-wise).

The entire alternating technique for two output channels is summarized in Alg. \ref{alg:blindideqeq_simow}. We denote this technique as \emph{alternating kernel canonical correlation analysis} (AKCCA).
Assuming $L<M$, the computational complexity of a single iteration of this algorithm is dominated by the KCCA problem. In particular, constructing its matrices and solving the GEV~\eqref{eq:blindideqgev_cca_reg} require $\mathcal{O}(NM)$ and $\mathcal{O}(M^3)$ time, respectively. If more than two output channels are present, the proposed algorithm follows exactly the same course, although the used formulae require to be extended, as shown in the sequel. A Matlab implementation of AKCCA can be obtained at \url{http://gtas.unican.es/people/steven}.

\begin{algorithm}[tbp]
\begin{algorithmic}
\STATE \textbf{input:} Output data sets $x_i[n]$ of the Wiener system.
\STATE Obtain the decomposed kernel matrices $\G_i$, see Eq.~\eqref{eq:decomposition}.
\STATE Center $\G_i$ with Eq.~\eqref{eq:centering}.
\STATE \textbf{initialize:} Set $\hat y_i[n] = x_i[n]$ and construct $\hat \Y_i$.
\REPEAT
\STATE{\textbf{CCA:} With given $\hat \Y_i$, update $\hat \h_i$ by solving Eq.~\eqref{eq:blindideqcca_h}.}%
\STATE{With new $\hat \h_i$, update $\W_i$ as in Eq.~\eqref{eq:def_w2_G}.}%
\STATE{\textbf{KCCA:} With given $\W_i$, update $\hat{\vec\alpha}_i$ by solving Eq.~\eqref{eq:blindideqkcca_alpha_reg}.}%
\STATE With new ${\hat{\vec\alpha}}_i$, update $\hat\y_i$ as in Eq.~\eqref{eq:yG} 
and construct $\hat\Y_i$. 
\UNTIL convergence
\STATE{Apply linear ZF or MMSE equalization on $\hat y_i[n]$ and $\hat \h_i$.}
\STATE \textbf{output:} $s[n]$
\end{algorithmic}
\caption{Alternating KCCA (AKCCA) for Blind Equalization of SIMO Wiener Systems} \label{alg:blindideqeq_simow}
\end{algorithm}

\begin{figure*}[tp]
\centering %
\includegraphics{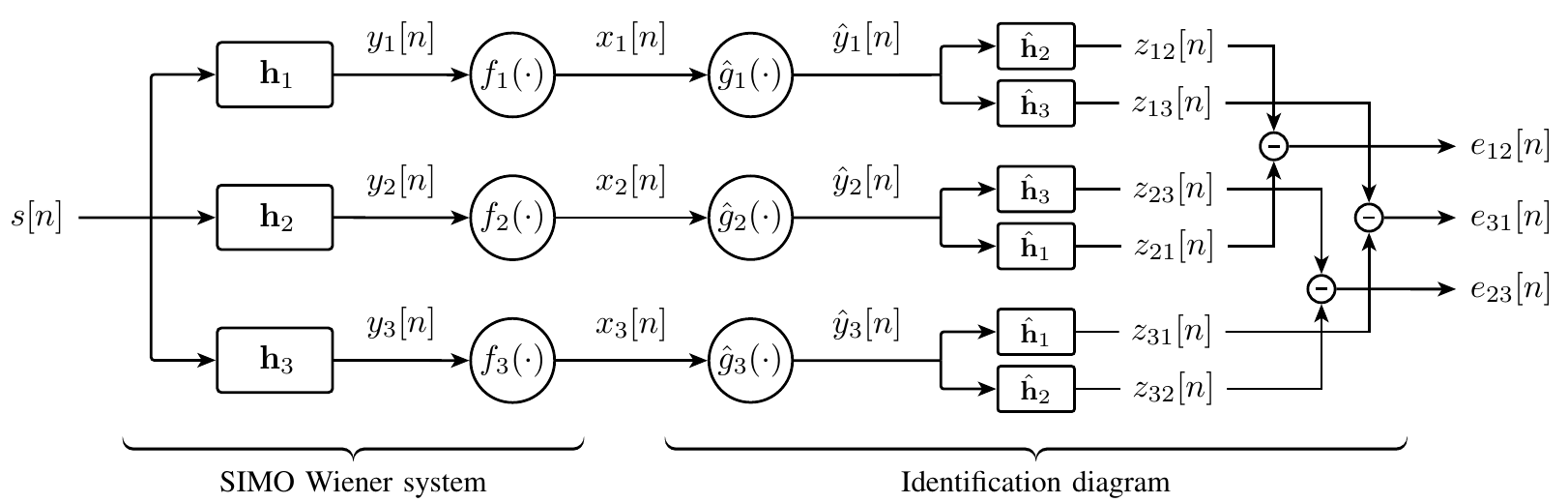}
\caption{A SIMO system consisting of three Wiener subsystems, and the proposed identification diagram.
} \label{fig:blindideqwiener_system_simo3_eq}
\end{figure*}

\section{Extensions} \label{sec:ext}

\subsection{Algorithm for systems with multiple outputs}  \label{sec:multi}

The proposed algorithm for systems with $2$ outputs can be extended to systems with an arbitrary number of outputs, say $P$, in a straightforward fashion. The problem \eqref{eq:optim2_simple} is extended from $2$ to $P$ outputs as
\begin{equation}
\begin{aligned}
& \underset{\hat{\vec\alpha}_1,\dots,\hat{\vec\alpha}_P, \hat\h_1,\dots,\hat\h_P}{\text{minimize}} & & \mathop{\sum_{i,j=1}^P}_{i\neq j}
\|\z_{ij}-\z_{ji}\|^2\\
& \text{subject to} & & \mathop{\sum_{i,j=1}^P}_{i\neq j} \|\z_{ij}\|^2 = 1,
\end{aligned}
\label{eq:optimP_simple}
\end{equation}
where $\z_{ij} = [z_{ij}[1], \dots, z_{ij}[N]]^T$ contains the signal $z_{ij}[n]$ obtained by transforming the output signal $x_i[n]$ by $\hat g_i(\cdot)$ and filtering it by
$\hat\h_j$. The identification diagram of Fig.~\ref{fig:blindideqwiener_system_simo3_eq} illustrates this optimization diagram for the case of $P=3$. While the energy restriction of problem \eqref{eq:optimP_simple} is slightly different compared to the restriction of the problem \eqref{eq:optim2_simple} for $2$ outputs, it was shown in \cite{Via06_effectiveorder}
that they are equivalent for these problems.

The optimization problem \eqref{eq:optimP_simple} can be solved by extending the iterative technique of Alg.~\ref{alg:blindideqeq_simow} to multiple channels. To this end, the problems \eqref{eq:blindideqcca_h} and \eqref{eq:blindideqkcca_alpha_reg} in Alg.~\ref{alg:blindideqeq_simow} require to be replaced by their multi-channel equivalents:

\subsubsection{Iteration 1}
Given estimates of ${\hat{\vec\alpha}}_i$, a set of new estimates of $\hat{\h}_i$ is found by solving
\begin{equation}
\begin{aligned}
& \underset{\hat\h_1,\dots,\hat\h_P}{\text{minimize}} & & \mathop{\sum_{i,j=1}^P}_{i\neq j} \| \hat \Y_i \hat \h_j - \hat \Y_j \hat \h_i \|^2\\
& \text{subject to} & & \mathop{\sum_{i,j=1}^P}_{i\neq j} \| \hat \Y_i \hat \h_j \|^2 = 1,
\end{aligned}
\label{eq:blindideqcca_h_P}
\end{equation}
where the elements of the matrices $\hat \Y_i$ are obtained through \eqref{eq:yG}.
The solution of the minimization problem \eqref{eq:blindideqcca_h_P} can be found as the principal eigenvector of the GEV
\begin{equation}
\R_{{\hat{\vec\alpha}}}\hat \h = \rho \D_{{\hat{\vec\alpha}}} \hat \h, \label{eq:blindideqgev_cca_P2}
\end{equation}
in which
\begin{equation}
\R_{{\hat{\vec\alpha}}} = \begin{bmatrix}
\0 & \hat\Y_2^T\hat\Y_1 & \cdots & \hat\Y_P^T\hat\Y_1\\
\hat\Y_1^T\hat\Y_2 & \0 & \cdots & \hat\Y_P^T\hat\Y_2\\
\vdots & \vdots & \ddots & \vdots\\
\hat\Y_1^T\hat\Y_P & \hat\Y_2^T\hat\Y_P & \cdots & \0
\end{bmatrix}, \label{eq:R_alpha_p}
\end{equation}
$\D_{{\hat{\vec\alpha}}}$ is a block-diagonal matrix whose $i$-th block on the diagonal is
$ \sum_{j=1;j\neq i}^P \hat\Y_j^T\hat\Y_j$, i.e.
\begin{equation}
\D_{{\hat{\vec\alpha}}} = \begin{bmatrix}
\sum_{j=2}^P \hat\Y_j^T\hat\Y_j & \cdots & \0\\
\vdots & \ddots & \vdots\\
\0 & \cdots & \sum_{j=1}^{P-1}\hat\Y_j^T\hat\Y_j
\end{bmatrix} \label{eq:d_hat_a}
\end{equation}
and $\hat \h$ contains the different estimated filters $\hat \h = [\hat \h_1^T, \hat \h_2^T, \dots, \hat \h_P^T]^T$.

\subsubsection{Iteration 2}
Subsequently, the estimates of $\hat{\h}_i$ are fixed and new estimates of ${\hat{\vec\alpha}}_i$ are obtained by solving
\begin{equation}
\begin{aligned}
& \underset{\hat{\vec\alpha}_1,\dots,\hat{\vec\alpha}_P}{\text{minimize}} & & \mathop{\sum_{i,j=1}^P}_{i\neq j}
\| \W_{ij} \hat{\vec\alpha}_i - \W_{ji} \hat{\vec\alpha}_j \|^2 \\
& \text{subject to} & & \mathop{\sum_{i,j=1}^P}_{i\neq j}
\| \W_{ij} \hat{\vec\alpha}_i \|^2 + c \sum_i^P \|\hat{\vec\alpha}_i\|^2 = 1
\end{aligned}
\label{eq:blindideqkcca_alpha_P}
\end{equation}
where the auxiliary variables $W_{ij}$ are defined as
\begin{equation*}
W_{ij}[n,m] = \sum_{l=0}^{L-1} \hat h_j[l] G_i[n-l,m].
\end{equation*}
Again, the minimization problem \eqref{eq:blindideqkcca_alpha_P} can be solved by retrieving the principal eigenvector of the corresponding GEV, which is found as
\begin{equation}
\R_{\hat{\h}} \hat{\vec\alpha} = \rho \D_{\hat\h} \hat{\vec\alpha}, \label{eq:blindideqgev_cca_P}
\end{equation}
in which
\begin{equation}
\R_{\hat{\h}} = \begin{bmatrix}
\0 & \W_{12}^T \W_{21} & \cdots & \W_{1P}^T \W_{P1}\\
\W_{21}^T \W_{12} & \0 & \cdots & \W_{2P}^T \W_{P2}\\
\vdots & \vdots & \ddots & \vdots\\
\W_{P1}^T \W_{1P} & \W_{P2}^T \W_{2P} & \cdots & \0 \label{eq:R_h_p}
\end{bmatrix},
\end{equation}
$\D_{\hat{\h}}$ is a regularized block-diagonal matrix whose $i$-th block on the diagonal is
$\sum_{j=1;j\neq i}^P \W_{ij}^T \W_{ij}$, i.e.
\begin{equation}
\D_{\hat{\h}} = \begin{bmatrix}
\sum_{j=2}^P \W_{1j}^T \W_{1j} & \cdots & \0\\
\vdots & \ddots & \vdots\\
\0 & \cdots & \sum_{j=1}^{P-1} \W_{Pj}^T \W_{Pj}
\end{bmatrix}+c\I \label{eq:d_hat_h}
\end{equation}
and $\hat{\vec\alpha} = [\hat{\vec\alpha}_1^T, \hat{\vec\alpha}_2^T, \dots, \hat{\vec\alpha}_P^T]^T$.

\subsection{Algorithm for systems with identical nonlinearities} \label{sec:identical}

In some cases it is known a priori that the $P$ nonlinearities $f_i(\cdot)$ are identical, for instance if the SIMO Wiener system is obtained by oversampling a SISO Wiener system. The validity of the oversampled model follows from the fact that the SISO system's nonlinearity is memoryless, and thus it applies to the signal $y[n]$ on a sample-by-sample basis. Therefore, it does not matter if one oversamples the internal signal $y[n]$ (similar to the linear case of Section \ref{sec:simolinear}) or the output signal $x[n]$.

The knowledge that the nonlinearities are identical can be exploited to obtain a more accurate estimate. Specifically, the data $\y = [\y_1^T, \dots, \y_P^T]^T$ can be estimated jointly as
\begin{equation*}
\hat\y = \G \hat{\vec\alpha},
\end{equation*}
where $\G = [\G_1^T, \dots, \G_P^T]^T$ is obtained by decomposing the kernel matrix of all data $[x_1[1],x_1[2], \dots, x_p[N]]^T$ and the vector $\hat{\vec\alpha} \in \mathbb{R}^{M\times 1}$ contains the expansion coefficients. The matrices $\R_{\hat \h}$ \eqref{eq:R_h_p} and $\D_{\hat\h}$ \eqref{eq:d_hat_h} in the GEV problem \eqref{eq:blindideqgev_cca_P} reduce to the $M\times M$ matrices
\begin{equation*}
\R_{\hat \h} = \mathop{\sum_{i,j=1}^P}_{i\neq j} \W_{i,j}^T\W_{j,i}
\end{equation*}
and
\begin{equation*}
\D_{\hat\h} = \mathop{\sum_{i,j=1}^P}_{i\neq j} \W_{i,j}^T\W_{i,j} + c\I.
\end{equation*}
We denote this extension of the algorithm as AKCCA-I.

\section{Experiments} \label{sec:experiments}

We now demonstrate the performance of the proposed algorithm through a number of computer simulations. Several different SIMO Wiener systems are used throughout these experiments. Their linear channels are taken from Table~\ref{table:impulse}, in which the impulse responses are chosen randomly, $h_i[j] \in \mathcal{N}(0,1)$. Their nonlinearities are chosen from the following monotonic invertible functions:
\begin{enumerate}
\item $f_1(y) = \tanh(0.8y)+0.1y$, a smooth saturation;
\item $f_2(y) = -0.1 \sin(3y) -0.33y $, a ``stairway'' function;
\item $f_3(y) = 1.5 y - 2.5\frac{1-\exp(-y)}{1+\exp(-y)}$, a smooth deadzone.
\end{enumerate}
The inverse functions of these nonlinearities can be observed in Fig.~\ref{fig:Ncomparison} of the first experiment.

\begin{table}
\caption{Impulse responses of the linear channels used in the simulations.}%
\newcommand{\zerowidth}[1]{\hbox to 0pt{\hss#1\hss}}
\setlength{\tabcolsep}{.9em}
\begin{tabular}[C]{ l r  r  r  r  r} \hlx{hv]}
\multicolumn{1}{c}{\zerowidth{$i$}} & %
\multicolumn{1}{c}{\zerowidth{$h_i[0]$}} & %
\multicolumn{1}{c}{\zerowidth{$h_i[1]$}} & %
\multicolumn{1}{c}{\zerowidth{$h_i[2]$}} & %
\multicolumn{1}{c}{\zerowidth{$h_i[3]$}} & %
\multicolumn{1}{c}{\zerowidth{$h_i[4]$}} \\ \hlx{vhv}
$1$ & $0.4115$ & $0.4165$ & $0.2249$ & $-0.0233$ & $-2.1971$\\
$2$ & $-0.5734$ & $0.1021$ & $-0.1259$ & $-0.4176$ & $0.6657$\\
$3$ & $1.4255$ & $0.6457$ & $-0.9509$ & $-0.1657$ & $-0.2512$\\
$4$ & $0.2846$ & $-0.3880$ & $0.5373$ & $0.7983$ & $0.4093$\\
$5$ & $-0.8769$ & $-0.3056$ & $-0.1160$ & $0.8130$ & $-0.8007$\\
\hlx{vhs}
\end{tabular}
\label{table:impulse}
\end{table}

The parameters of the AKCCA algorithm are set as follows: A Gaussian kernel is used with a different kernel width $\sigma_i$ for each channel. The kernel widths are chosen using Silverman's rule \cite[Section 3.4.2]{Silverman1986},
\begin{equation*}
\sigma = AN^{-\frac{1}{5}},
\end{equation*}
in which $N$ is the number of data points and $A = \min(d, (q_3-q_1)/1.34)$ is the minimum of the empirical standard deviation $d$ and the data interquartile range scaled by $1.34$.
The kernel matrix decompositions from Eq.~\eqref{eq:decomposition} are obtained by applying ICD \cite{Bach02kernelica} on the available data $\x_i[n]$. The precision of ICD is chosen as $10^{-8}$, resulting in values of $M$ within the range $10<M<50$ for all experiments. A standard regularization coefficient of $c=10^{-5}$ is fixed. Convergence of the AKCCA algorithm is assumed when the change in cost between two iterations is less than $10^{-10}$.

\subsection{Experiment 1: System identification}
In the first experiment we study the influence of the number of data, $N$, on the identification performance of the proposed algorithm. We also compare some results to a related supervised method.

An i.i.d. Gaussian signal is used as the input to a $1\times 3$ Wiener SIMO system, with linear channels $\h_1$, $\h_2$ and $\h_3$ (from Table~\ref{table:impulse}) and nonlinearities $f_1$, $f_2$ and $f_3$. No noise is assumed in first test. We perform system identification with AKCCA for input signals of three different sizes, $N=16$, $N=64$ and $N=256$.

The results of AKCCA are shown in Fig.~\ref{fig:Ncomparison}. Each column of plots shows the three estimated inverse nonlinearities $\hat g_i$ corresponding to one value for $N$, and the last column shows the estimated impulse responses of the linear channels for $N=256$. In order to account for the unknown scaling factor that is inherent to Wiener system identification, all estimates were scaled to obtain the same norm as their true values. While perfect system identification is only possible for source signals of infinite length, it is clear that by forcing smoothness onto the solution through a small amount of regularization the number of data to reach an acceptable solution is fairly low: Reasonable estimates are obtained for $N \geq 64$ in this experiment.

\begin{figure*}[t]
\centering
\begin{tabular}{cccc}
\includegraphics[width=0.2\linewidth]{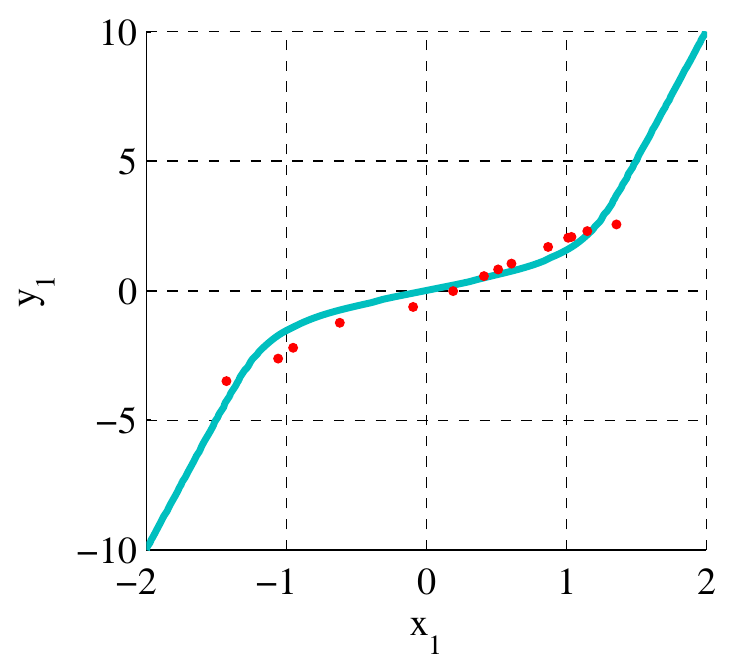} &
\includegraphics[width=0.2\linewidth]{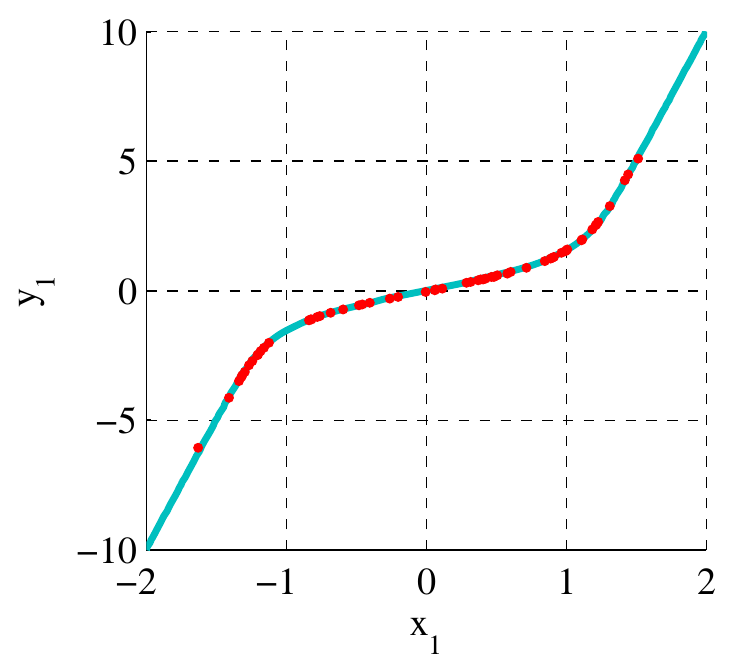} &
\includegraphics[width=0.2\linewidth]{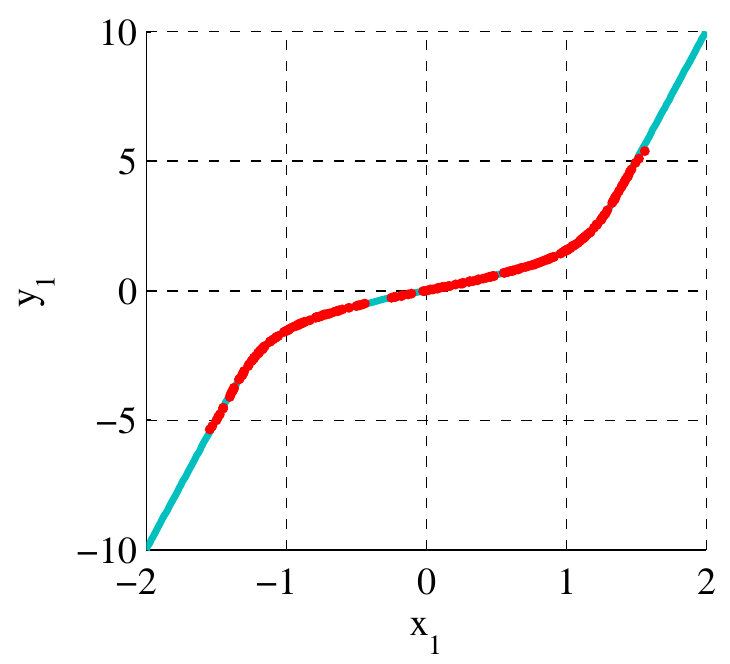} &
\includegraphics[width=0.2\linewidth]{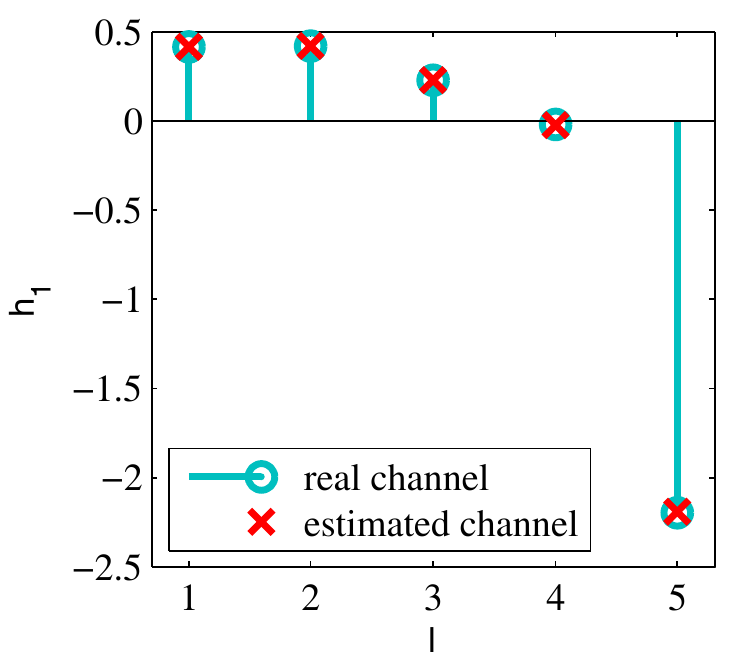} \\
\includegraphics[width=0.2\linewidth]{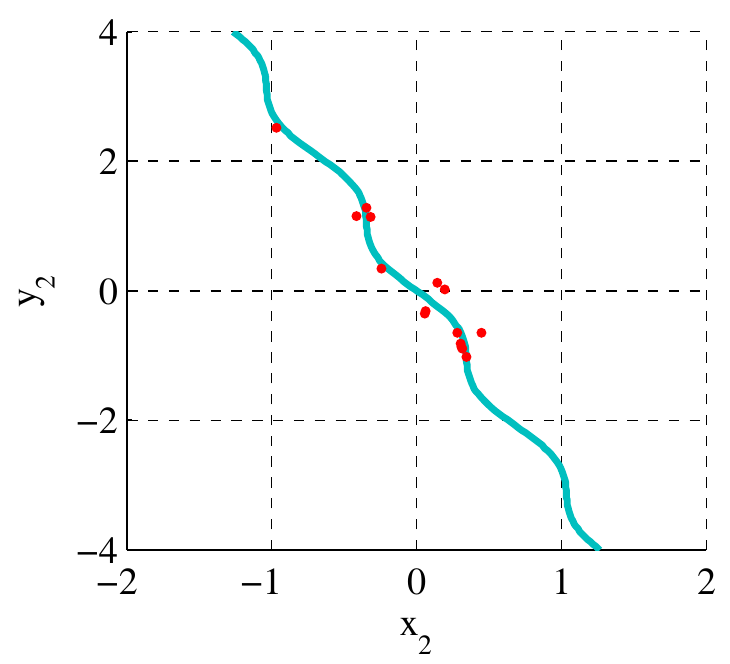} &
\includegraphics[width=0.2\linewidth]{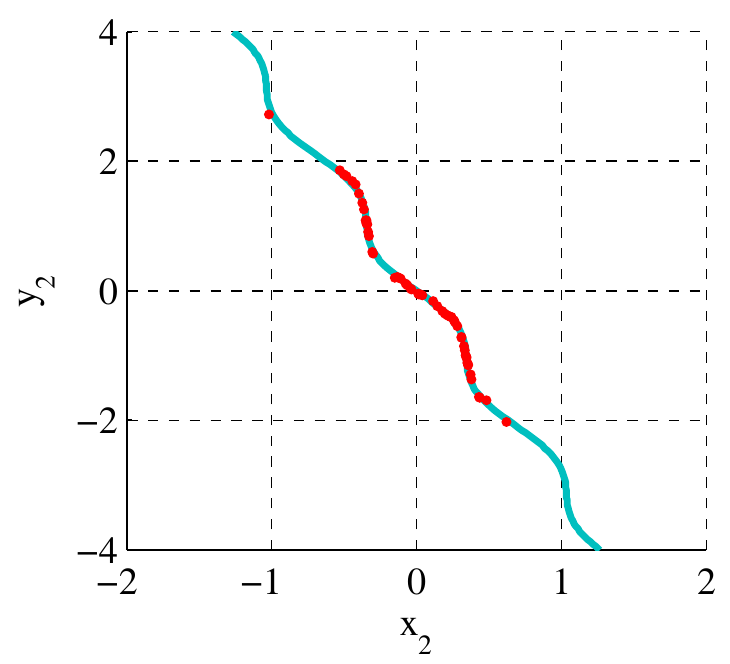} &
\includegraphics[width=0.2\linewidth]{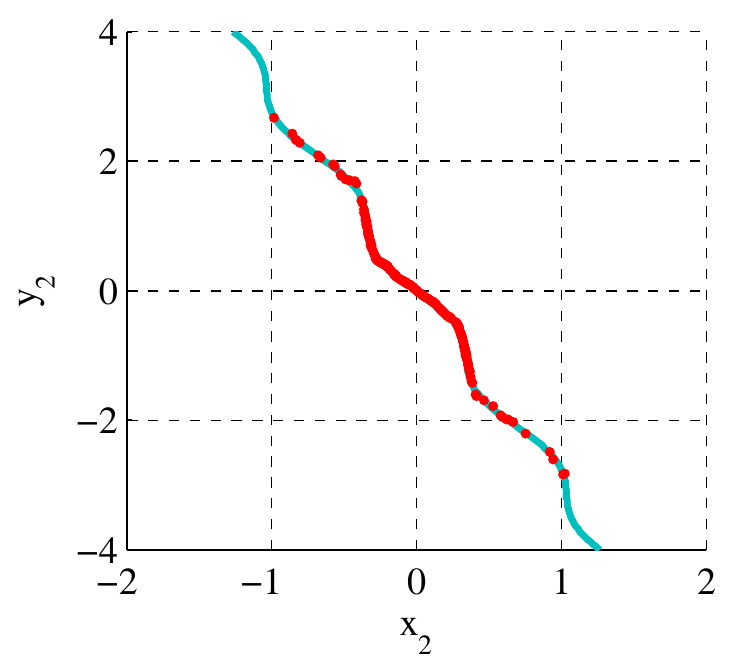} &
\includegraphics[width=0.2\linewidth]{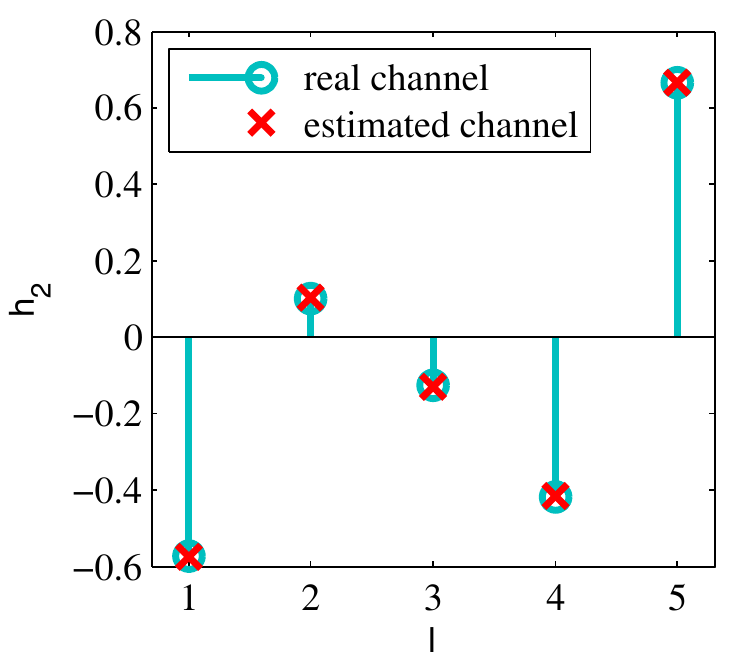} \\
\includegraphics[width=0.2\linewidth]{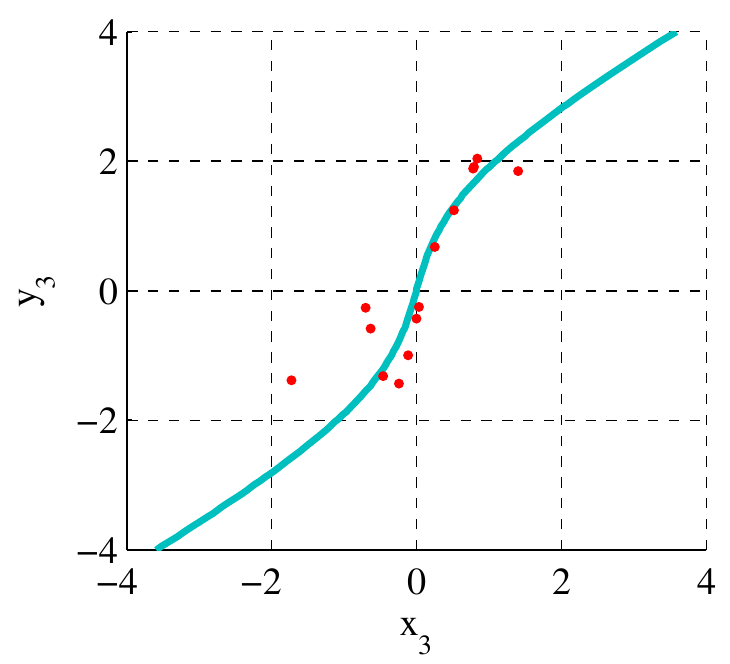} &
\includegraphics[width=0.2\linewidth]{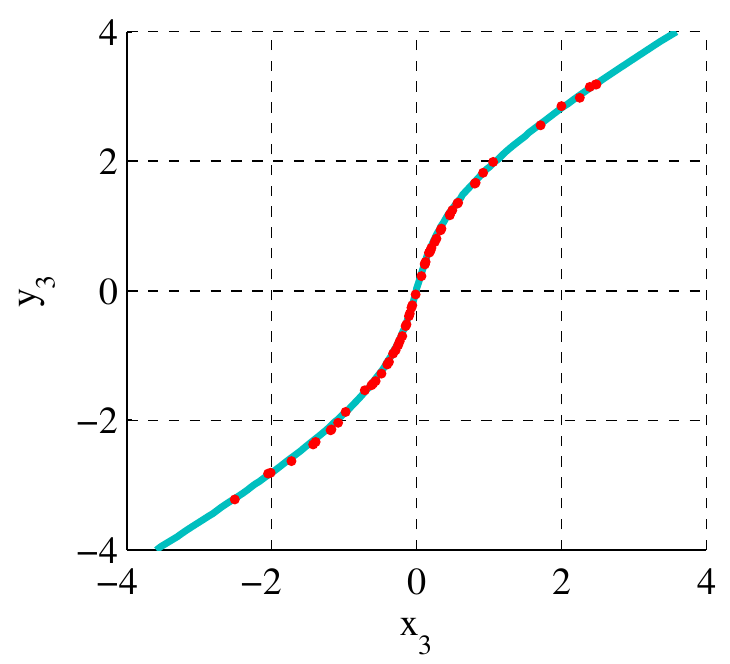} &
\includegraphics[width=0.2\linewidth]{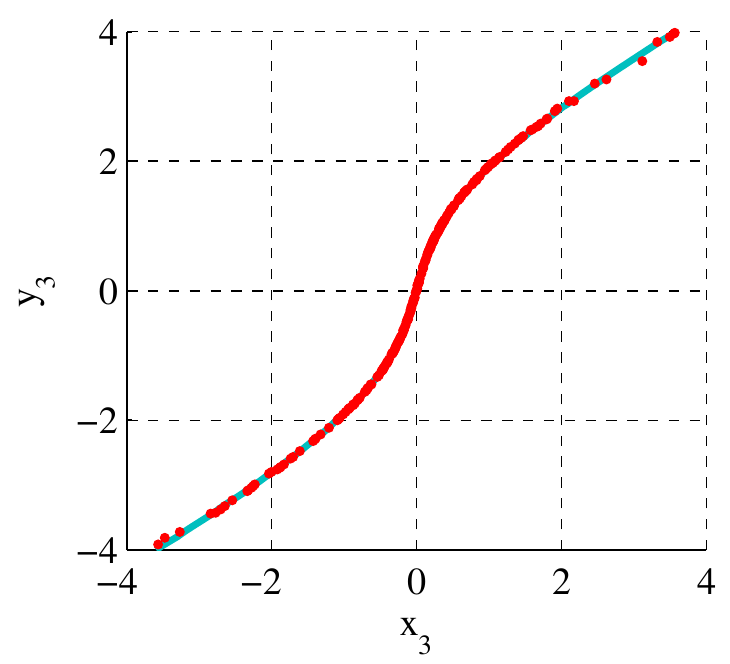} &
\includegraphics[width=0.2\linewidth]{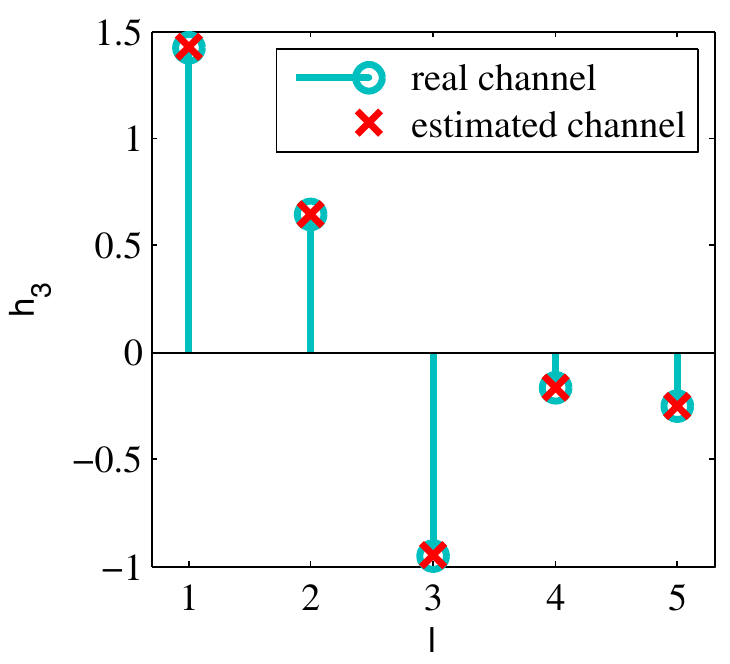} \\
$N=16$ & $N=64$ & $N=256$ & $N=256$ \\
\end{tabular}
\caption{Estimated inverse nonlinearities and linear channels for a $1\times 3$ Wiener system with different numbers of data $N$. Each row shows the estimates for one branch of the system; the solid line in the first three plots represents the true nonlinearity, and the dots indicate the true $x$ versus the estimated $y$-values.}
\label{fig:Ncomparison}
\end{figure*}

\begin{table}
\caption{Estimated impulse responses in experiment 1 for $N=256$. Top: AKCCA (blind). Bottom: supervised KCCA from \cite{vanvaerenbergh08_kcca_wiener_id}.}%
\newcommand{\zerowidth}[1]{\hbox to 0pt{\hss#1\hss}}
\setlength{\tabcolsep}{.9em}
\begin{tabular}[C]{l r r r r r} \hlx{hv]}
\multicolumn{1}{c}{\zerowidth{$i$}} & %
\multicolumn{1}{c}{\zerowidth{$h_i[0]$}} & %
\multicolumn{1}{c}{\zerowidth{$h_i[1]$}} & %
\multicolumn{1}{c}{\zerowidth{$h_i[2]$}} & %
\multicolumn{1}{c}{\zerowidth{$h_i[3]$}} & %
\multicolumn{1}{c}{\zerowidth{$h_i[4]$}} \\ \hlx{vhv}
$1$ & $0.4145$ & $0.4171$ & $0.2285$ & $-0.0235$ & $-2.1961$\\
$2$ & $-0.5740$ & $0.1034$ & $-0.1303$ & $-0.4153$ & $0.6655$\\
$3$ & $1.4271$ & $0.6466$ & $-0.9483$ & $-0.1655$ & $-0.2502$\\
\hlx{vhv}
$1$ & $0.4115$ & $0.4165$ & $0.2249$ & $-0.0232$ & $-2.1971$\\
$2$ & $-0.5742$ & $0.1019$ & $-0.1271$ & $-0.4151$ & $0.6663$\\
$3$ & $1.4256$ & $0.6458$ & $-0.9508$ & $-0.1657$ & $-0.2512$\\
\hlx{vhs}
\end{tabular}
\label{table:impulse_est}
\end{table}

Table~\ref{table:impulse_est} displays the impulse responses for $N=256$ as estimated by AKCCA. As a benchmark, we include the estimates obtained by the supervised KCCA-based identification algorithm method from \cite{vanvaerenbergh08_kcca_wiener_id} in the lower part of this table. This algorithm is performed in batch mode, on each subchannel individually. As can be observed, the performance of the proposed blind algorithm is fairly close to the performance of this related supervised technique. %

In order to study the convergence of the proposed AKCCA algorithm we plot its equalization MSE versus the number of iterations in Fig.~\ref{fig:convergence}. Equalization is carried out here by performing zero-forcing on the system identification result previously obtained. In addition to the noiseless scenario we also include results for the case of $20$ dB SNR. As can be observed, the algorithm typically converges in few iterations. For the noiseless case, convergence times on a $3$GHz $64$-bit Intel Core 2 PC with $4$ GB RAM running Matlab R2009b totaled respectively $0.34$s, $0.46$s, $0.72$s, $1.67$s and $3.05$s, for N ranging from $64$ to $1024$, as in Fig.~\ref{fig:convergence}.


\begin{figure}[t]
\begin{center}
\begin{tabular}{cc}
\includegraphics[width=10pc]{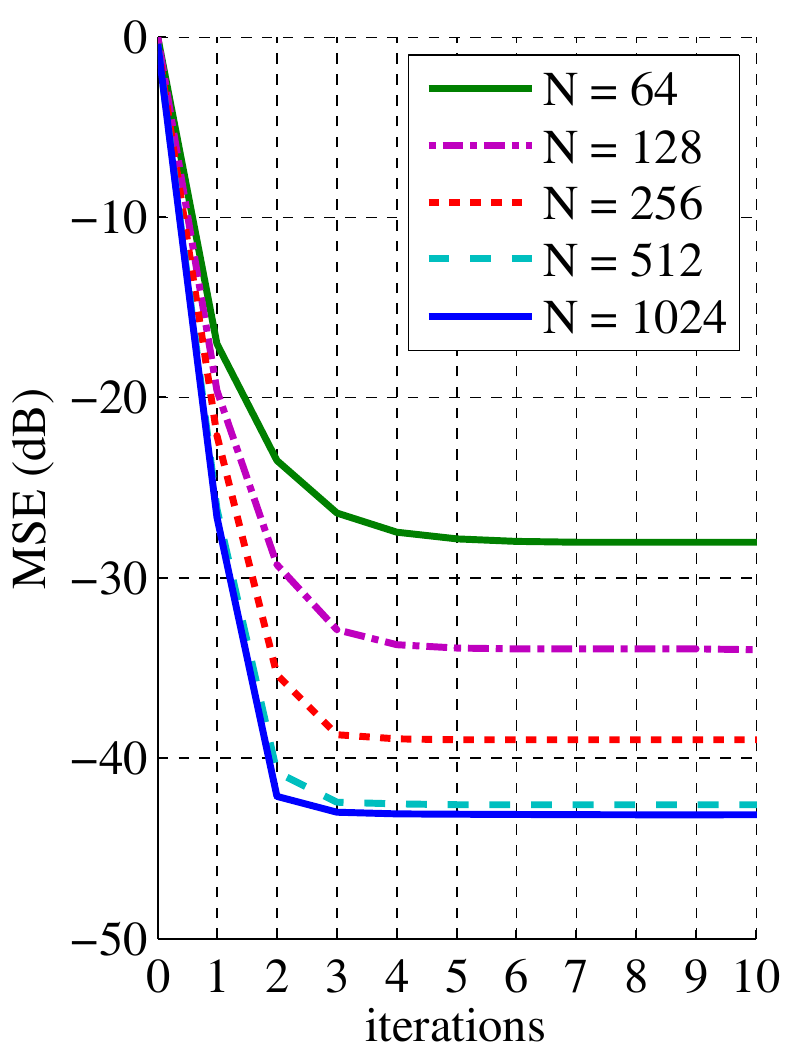} &
\includegraphics[width=10pc]{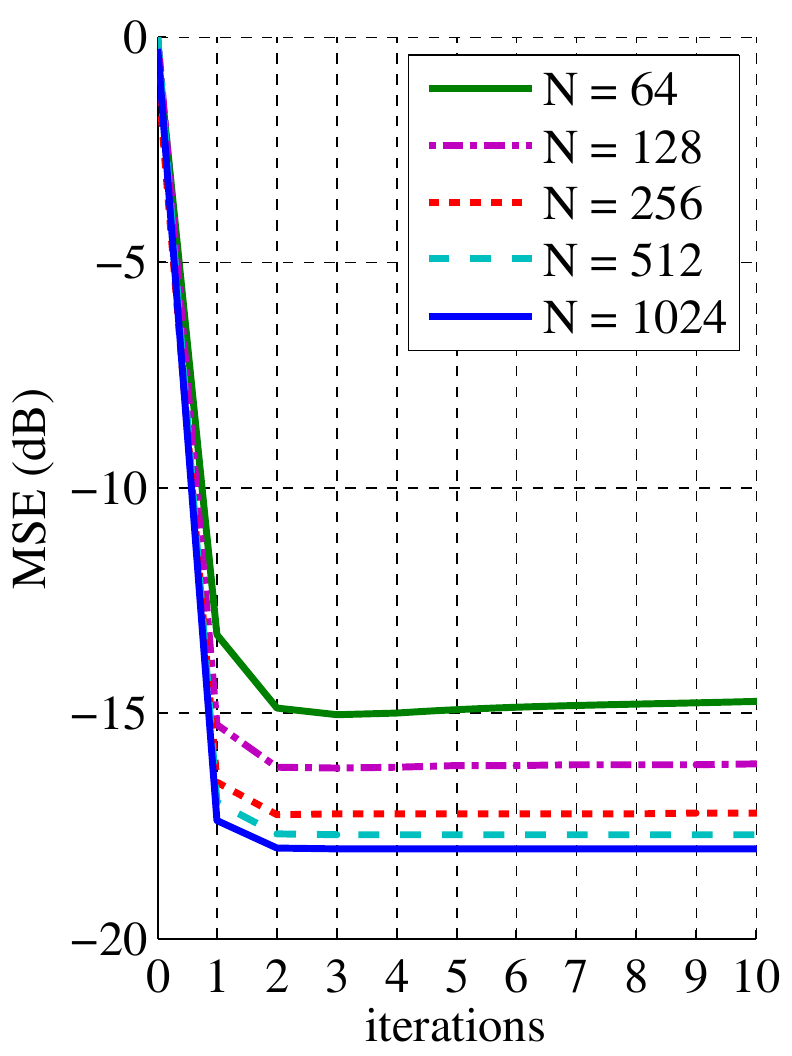}\\
(a) no noise & (b) $20$ dB SNR
\end{tabular}
\end{center}
\caption{Convergence of equalization MSE of AKCCA in experiment 1.} \label{fig:convergence}
\end{figure}

\subsection{Experiment 2: Comparison of equalization performance}
In the second experiment we examine the performance of the proposed algorithm on the problem of blind equalization of SIMO Wiener systems. The length of the source signal is fixed as $N=256$. The chosen SIMO Wiener system consists of the channels $\h_1$, $\h_2$ and $\h_3$ and an identical nonlinearity, $f_1$, in each branch. 
Since all SIMO branches share the same nonlinearity, we can compare 
the performance of the standard AKCCA algorithm to the performance of AKCCA-I, which exploits this property.
Different amounts of additive Gaussian white noise are added to the output of the system. In order to perform equalization, zero-forcing is performed after convergence is reached (see Algorithm \ref{alg:blindideqeq_simow}). We compare the equalization performance of the following algorithms:
\begin{enumerate}
\item {\bf CCA on linear SIMO system}: As a benchmark, we apply the blind linear CCA-based equalizer with zero-forcing from \cite{Via06_effectiveorder} on a system that only contains the linear channels $\h_1$, $\h_2$ and $\h_3$.
\item {\bf CCA on SIMO Wiener system}: The same blind linear method is applied to the chosen SIMO Wiener system.
\item {\bf AKCCA on SIMO Wiener system}: The proposed algorithm, applied to the SIMO Wiener system.
\item {\bf AKCCA-I on SIMO Wiener system}: The extension of the proposed algorithm that takes into account that the nonlinearities are identical (see Section \ref{sec:identical}).
\end{enumerate}
Fig.~\ref{fig:blindideqcca_comparison} shows the equalization MSE, calculated between the true input signal and the estimated input signal. Averages are taken over $100$ independent Monte-Carlo simulations. The results indicate that the proposed algorithms AKCCA and AKCCA-I show good overall performance, and AKCCA-I obtains an advantage over AKCCA at high SNR values starting at $40$ dB.

\begin{figure}[t]
\begin{center}
\includegraphics[width=21pc]{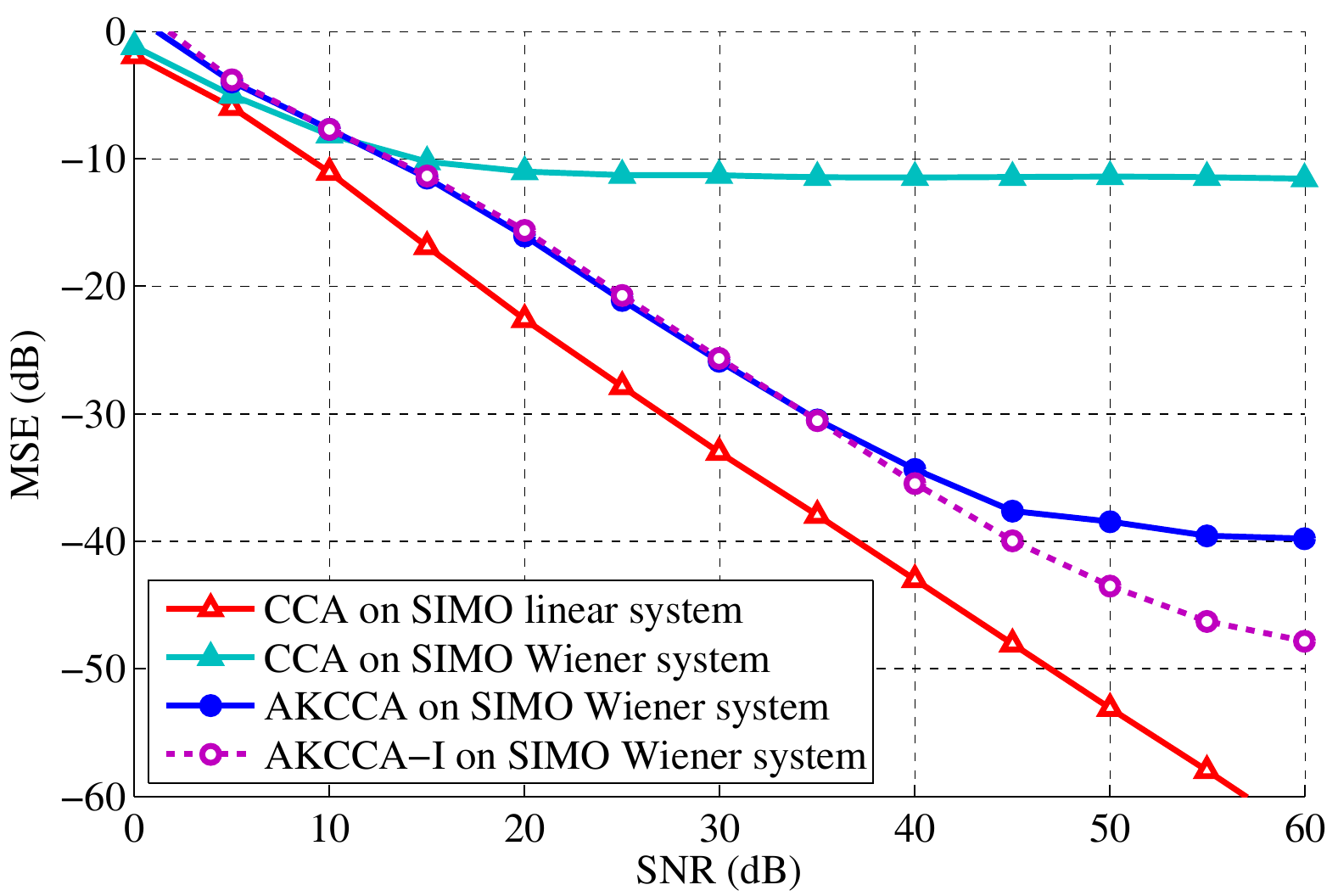}
\end{center}
\caption{Equalization MSE of CCA and KCCA-based algorithms on linear SIMO systems and SIMO Wiener systems.} \label{fig:blindideqcca_comparison}
\end{figure}

\subsection{Experiment 3: Influence of input characteristics and number of output channels}

In the last experiment we study the performance of the proposed AKCCA algorithm for different system configurations and input signal characteristics.

We perform three tests with different input signal characteristics. In each test we compare the performance on three systems with different numbers of output channels. These systems have $2$, $3$ and $4$ output channels, respectively, whose linear channels are taken in order from Table~\ref{table:impulse}, and whose nonlinearities are all chosen as $f_1$. The length of the source signal is fixed as $N=256$.

The availability of extra channels allows to exploit additional spatial diversity, and thus it is expected that a better result will be obtained. Nevertheless, when additional channels are available, the number of parameters to estimate is also higher, which raises the computational cost and could affect the results if only few data were available.

\subsubsection{Gaussian i.i.d. input signal}
First, we perform blind equalization with AKCCA on the three systems using a Gaussian i.i.d. input signal as in the previous experiments, i.e. $s[n] \in \mathcal{N}(0,1)$. The results for the final MSE after equalization are shown in Fig.~\ref{fig:example3_mse}. As can be observed, performance improves when channels are added, although the improvement per extra channel is smaller as more channels are added. In Table \ref{table:times} the average execution times are displayed, for the three systems and for different amounts of SNR. As expected, the algorithm requires more iterations to converge in noisy scenarios. Note also that the computational complexity scales cubically with the amount of subchannels of the system, see Section \ref{sec:overview}.

\begin{figure}[t]
\begin{center}
\includegraphics[width=21pc]{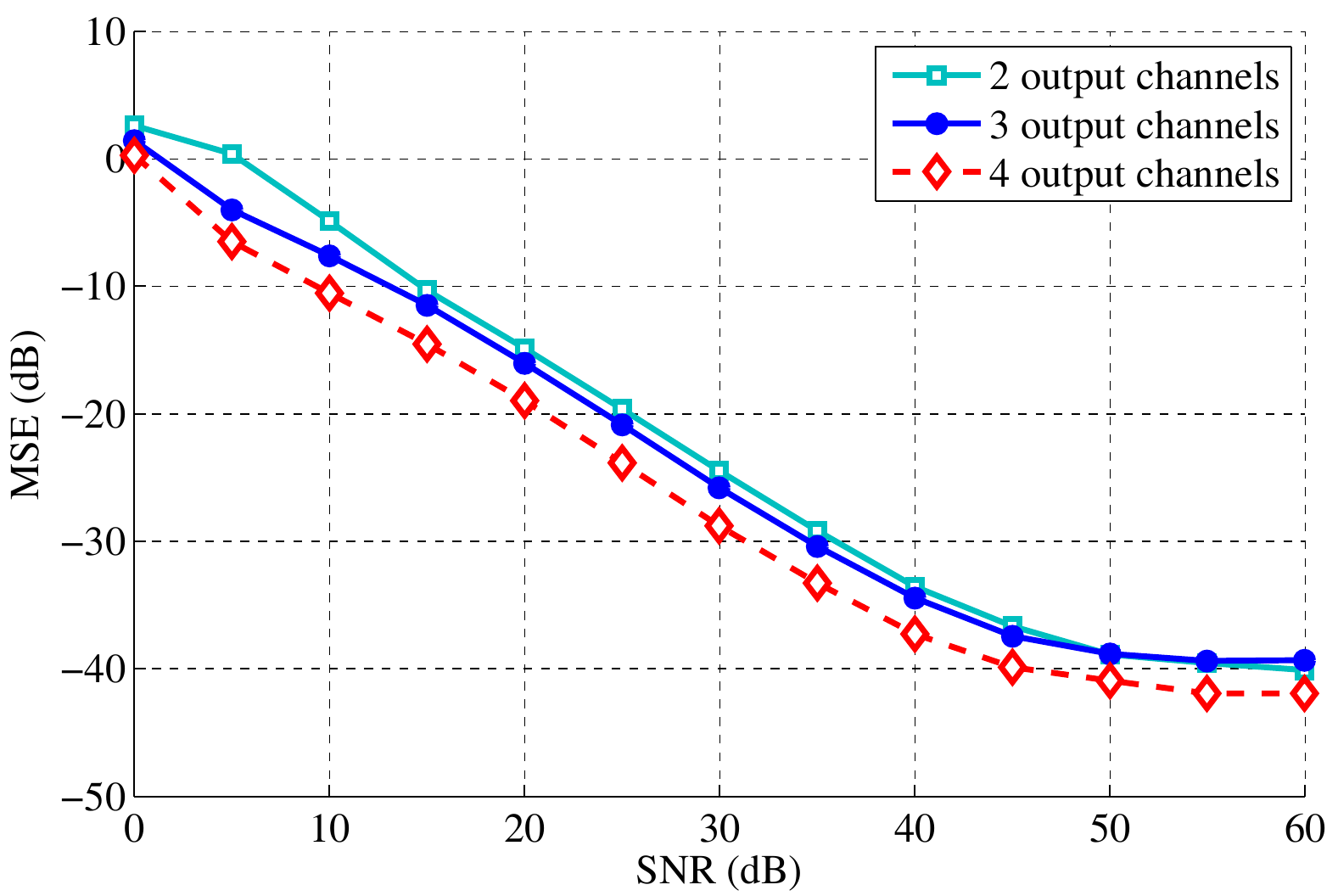}
\end{center}
\caption{Equalization MSE of AKCCA with a Gaussian i.i.d. input signal.} \label{fig:example3_mse}
\end{figure}

\begin{table}
\caption{Average execution times on a $3$GHz $64$-bit Intel Core 2 pc with $4$ GB RAM running Matlab R2009b.}%
\newcommand{\zerowidth}[1]{\hbox to 0pt{\hss#1\hss}}
\setlength{\tabcolsep}{.9em}
\begin{tabular}[C]{| r | c | c | c |} \hlx{hv]}
\multicolumn{1}{|c|}{\zerowidth{}} & %
\multicolumn{1}{c|}{$0$ dB SNR} & %
\multicolumn{1}{c|}{$30$ dB SNR} & %
\multicolumn{1}{c|}{$60$ dB SNR} \\ \hlx{vhv}
2 channels & $0.52$s & $0.14$s & $0.11$s\\
3 channels & $1.25$s & $0.40$s & $0.27$s\\
4 channels & $2.89$s & $0.84$s & $0.63$s\\
\hlx{vhs}
\end{tabular}
\label{table:times}
\end{table}

\subsubsection{Colored input signal}
In a second test we study the performance of AKCCA when the input signal is colored. In order to color the source signal $s[n]$ before it enters the SIMO Wiener system, we apply a $20$-tap low-pass FIR filter $\h_c$ onto it, i.e. $s'[n] = h_c * s[n]$, with cut-off frequency at $0.7\pi$ radians per sample and a stopband attenuation of $60$dB.
AKCCA is then performed to retrieve the filtered input signal $s'[n]$. As the linear channels $\h_i$ are now hardly excited in the stopband frequency range, it is much harder or even impossible to identify their complete frequency responses. Nevertheless, it may still be possible to estimate the colored input signal. The equalization MSE of AKCCA is shown in Fig.~\ref{fig:example4_mse}. Interestingly, the results are only slightly affected w.r.t. to the previous test (see Fig.~\ref{fig:example3_mse}), which demonstrates that the proposed method is also suitable for colored source signals, up to some extent.

\begin{figure}[t]
\begin{center}
\includegraphics[width=21pc]{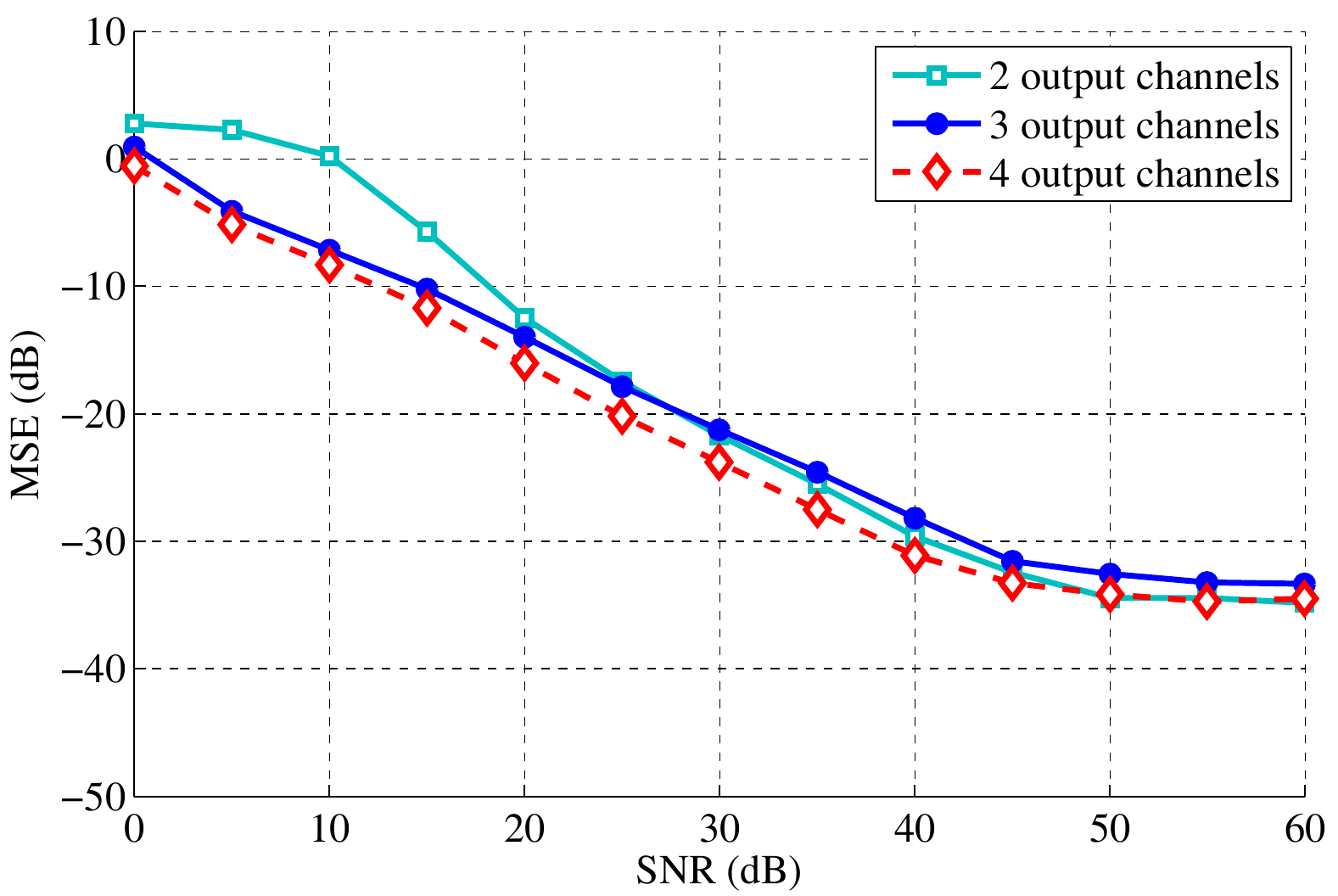}
\end{center}
\caption{Equalization MSE of AKCCA on SIMO Wiener systems with a colored input signal.} \label{fig:example4_mse}
\end{figure}

\subsubsection{Binary input signal}

Finally, the test is repeated on a system with a binary input, $s[n] \in \{-1,+1\}$. The obtained BER values are shown in Fig.~\ref{fig:example5_ber}. We include the results for a system that uses a fifth channel as an additional subchannel, $\h_5$ in this case. As can be observed, the performance of the two-channel system is substantially improved by adding a third subchannel. By including additional subchannels, furthermore, the performance keeps improving, though slightly less for each channel added.



\begin{figure}[t]
\begin{center}
\includegraphics[width=21pc]{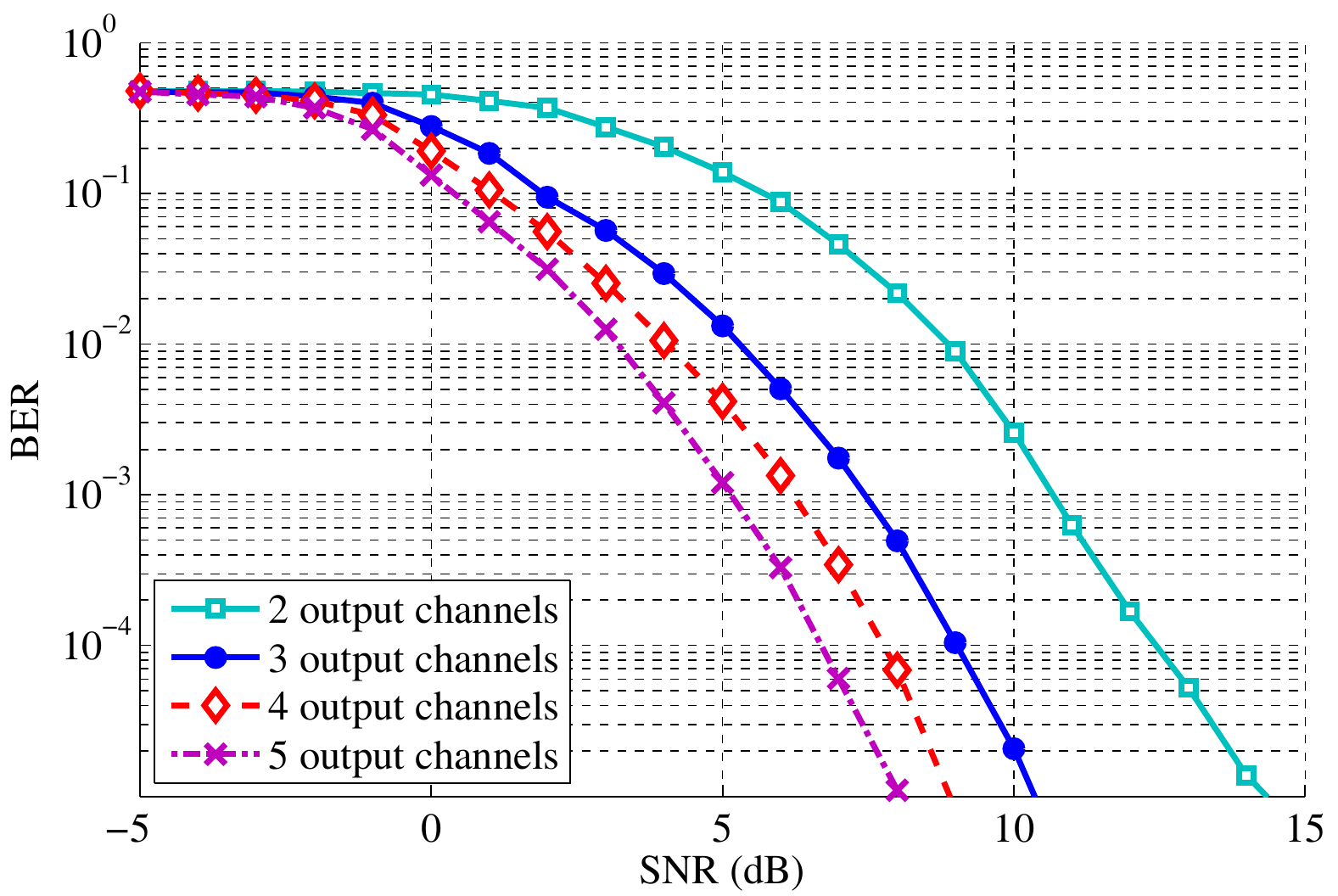}
\end{center}
\caption{Equalization BER of AKCCA on SIMO Wiener systems with a binary i.i.d. input signal.} \label{fig:example5_ber}
\end{figure}


\section{Conclusions}
\label{sec:conclusions}

We have considered the problems of blind identification and blind equalization of SIMO Wiener systems. These systems are multi-channel Wiener systems that are excited by a single shared source signal. By applying a kernel transformation to the output data, we have obtained a linearized optimization problem for which we have proposed an iterative algorithm. The proposed KCCA algorithm iterates between a CCA algorithm to estimate the linear channels and a KCCA algorithm to estimate the memoryless nonlinearities. While we have indicated that blind identification of SIMO Wiener systems is not possible in general for finite source signals, we have also shown that identifiability becomes possible in practice by posing some smoothness constraints on the nonlinearities.

We have provided a general formulation of the proposed technique that allows to operate on systems with two or more output channels. It performs well on reasonably small data sets and can handle systems with colored input signals. Results also show that it has fast convergence and it achieves identification performance that is very close to a related supervised method.

Several directions for future research are open. First, it would be interesting to perform a theoretical analysis in the presence of noise on this problem. Also, some of the ideas proposed in this paper may be applied to other block-based systems such as multiple-input multiple-output (MIMO) Wiener systems and configurations with Hammerstein systems.

\section*{Acknowledgment}
The authors would like to thank the anonymous reviewers for their suggestions on the identifiability of SIMO Wiener systems, and for their recommendations of the choice of the linear channels. This work was supported by MICINN (Spanish Ministry for Science and Innovation) under grants 
TEC2010-19545-C04-03 (COSIMA) and CONSOLIDER-INGENIO 2010 CSD2008-00010 (COMONSENS).

\bibliographystyle{IEEEtran}
\bibliography{biblio}

\begin{IEEEbiographynophoto}
{Steven Van Vaerenbergh (S'06–M'11)} received the M.Sc. degree in Electrical Engineering from Ghent University, Ghent, Belgium, in 2003, and the Ph.D. degree from the University of Cantabria, Santander, Spain, in 2010. He has been a Visiting Researcher with the Computational NeuroEngineering Laboratory, University of Florida, Gainesville. Currently, Dr. Van Vaerenbergh is a post-doctoral associate with the Department of Telecommunication Engineering, University of Cantabria, Spain. His main research interests include machine learning and its applications to adaptive filtering, target tracking and system identification.
\end{IEEEbiographynophoto}

\begin{IEEEbiographynophoto}
{Javier V\'ia (M'08-SM'12)} received his Telecommunication Engineer Degree and his Ph.D. in electrical engineering from the University of Cantabria, Spain in 2002 and 2007, respectively. In 2002 he joined the Department of Communications Engineering, University of Cantabria, Spain, where he is currently Associate Professor. He has spent visiting periods at the Smart Antennas Research Group of Stanford University, at the Department of Electronics and Computer Engineering (Hong Kong University of Science and Technology), and at the State University of New York at Buffalo. Dr. V\'ia has actively participated in several European and Spanish research projects. His current research interests include blind channel estimation and equalization in wireless communication systems, multivariate statistical analysis, quaternion signal processing and kernel methods.
\end{IEEEbiographynophoto}

\begin{IEEEbiographynophoto}
{Ignacio Santamar\'ia (M'96, SM'05)} received his Telecommunication Engineer Degree and his Ph.D. in electrical engineering from the Universidad Polit\'ecnica de Madrid (UPM), Spain, in 1991 and 1995, respectively. In 1992 he joined the Department of Communications Engineering, University of Cantabria, Spain, where he is currently Full Professor. He has more than 150 publications in refereed journals and international conference papers and holds 2 patents. His current research interests include signal processing algorithms for multi-user multi-antenna wireless communication systems, multivariate statistical techniques and machine learning theories. He has been involved in numerous national and international research projects on these topics. He has been a visiting researcher at the Computational NeuroEngineering Laboratory (University of Florida), and at the Wireless Networking and Communications Group (University of Texas at Austin).

Dr. Santamar\'ia was a Technical Co-Chair of the 2nd International ICST Conference on Mobile Lightweight Wireless Systems (MOBILIGHT 2010), Special Sessions Co-Chair of the 2011 European Signal Processing Conference (EUSIPCO 2011), and General Co-Chair of the 2012 IEEE Workshop on Machine Learning for Signal Processing (MLSP 2012). Since 2009 he has been a member of the IEEE Machine Learning for Signal Processing Technical Committee. Currently, he serves as Associate Editor of the IEEE Transactions on Signal Processing. He was a co-recipient of the 2008 EEEfCOM Innovation Award, as well as coauthor of a paper that received the 2012 IEEE SPS Young Author Best Paper Award.
\end{IEEEbiographynophoto}

\end{document}